\documentclass{article}

\PassOptionsToPackage{round}{natbib}
\usepackage{nips_like_2017}

\usepackage[utf8]{inputenc} 
\usepackage[T1]{fontenc}    
\usepackage{url}            
\usepackage{booktabs}       
\usepackage{amsfonts}       
\usepackage{nicefrac}       
\usepackage{microtype}      

\usepackage{amsfonts,amsmath,amssymb,amsthm}
\usepackage{verbatim,float,url,dsfont}
\usepackage{graphicx,subfigure,psfrag}
\usepackage{algorithm,algorithmic}
\usepackage{mathtools}
\usepackage[colorlinks=true,citecolor=blue,urlcolor=blue,linkcolor=blue]{hyperref}

\newtheorem{theorem}{Theorem}
\newtheorem{lemma}{Lemma}

\theoremstyle{definition}
\newtheorem{remark}{Remark}

\newcommand{\argmin}{\mathop{\mathrm{argmin}}}

\newcommand{\st}{\mathop{\mathrm{subject\,\,to}}}
\def\R{\mathbb{R}}

\def\half{\frac{1}{2}}

\def\col{\mathrm{col}}

\def\nul{\mathrm{null}}

\def\sign{\mathrm{sign}}
\def\supp{\mathrm{supp}}

\def\Id{\mathrm{Id}}
\def\setint{\mathrm{int}}

\title{Dykstra's Algorithm, ADMM, and Coordinate Descent: Connections,
  Insights, and Extensions}

\author{
Ryan J. Tibshirani \\
Statistics Department and Department of Machine Learning \\
Carnegie Mellon University \\
Pittsburgh, PA 15213 \\
\texttt{ryantibs@stat.cmu.edu} 
}

\begin{document}
\maketitle

\begin{abstract}
We study connections between Dykstra's algorithm for projecting onto
an intersection of convex sets, the augmented Lagrangian method of
multipliers or ADMM, and block coordinate descent.  We prove that  
coordinate descent for a regularized regression problem, in which the
(separable) penalty functions are seminorms, is exactly equivalent to   
Dykstra's algorithm applied to the dual problem.  ADMM on the dual
problem is also seen to be equivalent, in the special case of two
sets, with one 
being a linear subspace. These connections, aside from being
interesting in their own right, suggest new ways of analyzing and
extending coordinate descent. For example, from existing convergence
theory on Dykstra's algorithm over polyhedra, we discern that
coordinate descent for the lasso problem converges at an
(asymptotically) linear rate.  We also develop two parallel versions 
of coordinate descent, based on the Dykstra and ADMM connections.        
\end{abstract}

\section{Introduction}
\label{sec:intro}

In this paper, we study two seemingly unrelated but closely connected  
convex optimization problems, and associated algorithms.
The first is the {\it best approximation problem}: given
closed, convex sets $C_1,\ldots,C_d \subseteq \R^n$ and $y \in
\R^n$, we seek the point in the intersection $C_1 \cap \cdots \cap
C_d$ closest to $y$, and solve
\begin{equation}
\label{eq:bap}
\min_{u \in \R^n} \; \|y-u\|_2^2 
\quad \st \quad u \in C_1 \cap \cdots \cap C_d.
\end{equation}
The second problem is the {\it regularized regression problem}: given
a response $y \in \R^n$ and predictors $X \in \R^{n\times
  p}$, and a block decomposition $X_i \in \R^{n \times p_i}$,
$i=1,\ldots,d$ of the columns of $X$ (i.e., these could be columns, or
groups of columns), we build a working linear model by
applying blockwise regularization over the coefficients, and solve 
\begin{equation}
\label{eq:reg}
\min_{w \in \R^p} \; \half\|y-Xw\|_2^2 + \sum_{i=1}^d h_i(w_i), 
\end{equation}
where $h_i:\R^{p_i} \to \R$, $i=1,\ldots,d$ are convex functions, and
we write $w_i \in \R^{p_i}$, $i=1,\ldots,d$ for the appropriate block  
decomposition of a coefficient vector $w \in \R^p$ (so that
\smash{$Xw = \sum_{i=1}^d X_iw_i$}).

Two well-studied algorithms for problems \eqref{eq:bap},
\eqref{eq:reg} are {\it Dykstra's algorithm}
\citep{dykstra1983algorithm,boyle1986method} 
and {\it (block) coordinate descent} \citep{warga1963minimizing, 
bertsekas1989parallel,tseng1990dual}, respectively.  The jumping-off
point for our work in this paper is the following fact: {\it these two 
  algorithms are equivalent for solving \eqref{eq:bap} and
  \eqref{eq:reg}}.  That is, for a 
particular relationship between the sets $C_1,\ldots,C_d$ and
penalty functions $h_1,\ldots,h_d$, the problems \eqref{eq:bap} and
\eqref{eq:reg} are duals of each other, and Dykstra's algorithm on the
primal problem \eqref{eq:bap} is exactly the same as coordinate
descent on the dual problem \eqref{eq:reg}.  We provide details in
Section \ref{sec:connections}.    

This equivalence between Dykstra's algorithm and
coordinate descent can be essentially found in the optimization
literature, dating back to the late 1980s, and possibly
earlier.  (We say ``essentially'' here because, to our knowledge, this
equivalence has not been stated for a general regression matrix $X$,  
and only in the special case $X=I$; but, in truth, the extension
to a general matrix $X$ is fairly straightforward.)  Though this 
equivalence has been cited and discussed in various ways 
over the years, we feel that it is not as well-known as it
should be, especially in light of the recent resurgence
of interest in coordinate descent methods.  We revisit
the connection between Dykstra's algorithm and coordinate descent, and  
draw further connections to a third method---the {\it augmented 
Lagrangian method of multipliers} or ADMM
\citep{glowinski1975approximation,gabay1976dual}---that has 
also received a great deal of attention recently.  While these basic
connections are interesting in their own right, 
they also have important implications for analyzing and extending  
coordinate descent. Below we give a summary of our contributions.   

\begin{enumerate}
\item We prove in Section
  \ref{sec:connections}  (under a particular configuration 
  relating $C_1,\ldots,C_d$ to seminorms $h_1,\ldots,h_d$)
  that Dykstra's algorithm for \eqref{eq:bap} is equivalent to block 
  coordinate 
  descent for \eqref{eq:reg}.  (This is a mild
  generalization of the previously known connection when $X=I$.)   
\item We also show in Section \ref{sec:connections} that ADMM is 
  closely connected to Dykstra's algorithm, in that ADMM for
  \eqref{eq:bap}, when $d=2$ and $C_1$ is a linear subspace, matches
  Dykstra's algorithm. 
\item Leveraging existing results on the convergence of Dykstra's
  algorithm for an intersection of halfspaces, we establish in Section 
  \ref{sec:lasso} that coordinate descent for the lasso problem has an
  (asymptotically) linear rate of convergence, regardless of the
  dimensions of $X$ 
  (i.e., without assumptions about strong convexity of the problem).
  We derive two different explicit forms for the error constant, which
  shed light onto how correlations among the predictor variables
  affect the speed of convergence.
\item Appealing to parallel versions of Dykstra's algorithm and
  ADMM, we present in Section \ref{sec:parallel} two parallel
  versions of coordinate descent (each guaranteed to converge in full 
  generality).
\item We extend in Section \ref{sec:extensions} the equivalence
  between coordinate descent and Dykstra's algorithm to the case
  of nonquadratic loss in \eqref{eq:reg}, i.e., non-Euclidean
  projection in \eqref{eq:bap}. This leads to a
  Dykstra-based parallel version of coordinate descent for (separably 
  regularized) problems with nonquadratic loss, and we also derive an
  alternative ADMM-based parallel version of coordinate descent for
  the same class of problems.
\end{enumerate}

\section{Preliminaries and connections}
\label{sec:connections}

\paragraph{Dykstra's algorithm.} 
Dykstra's algorithm was first proposed by
\citet{dykstra1983algorithm}, and was extended to Hilbert spaces by 
\citet{boyle1986method}.  Since these seminal papers, a number of
works have analyzed and extended Dykstra's algorithm 
in various interesting ways. We will reference many of these
works in the coming sections, when we discuss connections
between Dykstra's  
algorithm and other methods; for other developments, see the  
comprehensive books \citet{deutsch2001best,bauschke2011convex} and  
review article \citet{bauschke2013projection}. 

Dykstra's algorithm for the best approximation problem \eqref{eq:bap}
can be described as follows.  We initialize \smash{$u^{(0)}=y$},
\smash{$z^{(-d+1)}=\cdots=z^{(0)}=0$}, and then repeat, for 
$k=1,2,3,\ldots$: 
\begin{equation}
\begin{aligned}
\label{eq:dyk}
u^{(k)} &= P_{C_{[k]}} (u^{(k-1)}+z^{(k-d)}), \\
z^{(k)} &= u^{(k-1)} + z^{(k-d)}-u^{(k)},
\end{aligned}
\end{equation}
where \smash{$P_C(x)=\argmin_{c \in C} \|x-c\|_2^2$} denotes the 
(Euclidean) projection of $x$ onto a closed, convex set $C$, and 
$[\cdot]$ denotes the modulo operator taking values in 
$\{1,\ldots,d\}$. What differentiates Dykstra's   
algorithm from the classical {\it alternating projections method} of 
\citet{vonneumann1950functional,halperin1962product} is the sequence
of (what we may call) dual variables \smash{$z^{(k)}$},
$k=1,2,3,\ldots$.  These track, in a cyclic fashion, the residuals
from projecting onto $C_1,\ldots,C_d$.  The simpler alternating  
projections method will always converge to a feasible point in
$C_1 \cap \cdots \cap C_d$, but will not necessarily
converge to the solution in \eqref{eq:bap} unless $C_1,\ldots,C_d$ are 
subspaces (in which case alternating projections and Dykstra's
algorithm coincide).  Meanwhile, Dykstra's algorithm converges in
general (for any closed, convex sets $C_1,\ldots,C_d$ with
nonempty intersection, see, e.g., 
\citet{boyle1986method,han1988successive,gaffke1989cyclic}).       
We note that Dykstra's algorithm \eqref{eq:dyk}
can be rewritten in a different form, which will be
helpful for future comparisons. First, we initialize
\smash{$u_d^{(0)}=y$}, \smash{$z_1^{(0)}=\cdots=z_d^{(0)}=0$}, and
then repeat, for $k=1,2,3,\ldots$:  
\begin{equation}
\begin{aligned}
\label{eq:dyk2}
&u_0^{(k)} = u_d^{(k-1)}, \\
&\begin{rcases*}
u_i^{(k)} = P_{C_i} (u_{i-1}^{(k)} + z_i^{(k-1)}), & \\
z_i^{(k)} = u_{i-1}^{(k)} + z_i^{(k-1)} - u_i^{(k)}, &
\end{rcases*}
\quad \text{for $i=1,\ldots,d$}. 
\end{aligned}
\end{equation}

\paragraph{Coordinate descent.}
Coordinate descent methods have a long history in
optimization, and have been studied and discussed in early papers 
and books such as \citet{warga1963minimizing,ortega1970iterative,  
luenberger1973introduction,auslender1976optimisation,
bertsekas1989parallel}, though coordinate descent was still likely
in use much earlier. 
(Of course, for solving linear systems, coordinate descent reduces to
Gauss-Seidel iterations, which dates back to the 1800s.)  
Some key papers analyzing the convergence of coordinate descent
methods are \citet{tseng1987relaxation,tseng1990dual,
luo1992convergence,luo1993convergence,tseng2001convergence}.  In the
last 10 or 15 years, a considerable interest in coordinate descent has 
developed across the optimization community. With the   
flurry of recent work, it would be difficult to give a thorough 
account of the recent progress on the topic. To give just a few
examples, recent developments include finite-time
(nonasymptotic) convergence rates for coordinate descent, and exciting 
extensions such as accelerated, parallel, and distributed versions of
coordinate descent.  We refer to \citet{wright2015coordinate}, an
excellent survey that describes this recent progress.  

In (block) coordinate descent\footnote{To be precise, this is {\it
  cyclic} coordinate descent, where {\it exact} minimization is
performed along each block of coordinates. {\it Randomized} versions
of this algorithm have recently become popular,
as have {\it inexact} or {\it proximal} versions.
While these variants are
interesting, they are not the focus of our paper.} for
\eqref{eq:reg}, we initialize say \smash{$w^{(0)}=0$}, and repeat, for 
$k=1,2,3,\ldots$: 
\begin{equation}
\label{eq:cd}
w_i^{(k)} = \argmin_{w_i \in \R^{p_i}} \; 
\half \bigg\|y - \sum_{j < i} X_j w_j^{(k)} - \sum_{j > i} X_j
w_j^{(k-1)} - X_i w_i \bigg\|_2^2 + h_i (w_i),
\quad i=1,\ldots,d. 
\end{equation}
We assume here and throughout that 
$X_i \in \R^{n \times p_i}$, $i=1,\ldots,d$ each have full 
column rank so that the updates in \eqref{eq:cd} are uniquely defined 
(this is used for convenience, and is not a strong 
assumption; note that it places no restriction on the dimensionality
of the full problem in \eqref{eq:reg}, i.e., we could still have $X
\in \R^{n\times p}$ with $p \gg n$). 
The precise form of these updates, of course,
depends on the penalty functions $h_i$, $i=1,\ldots,d$.  Suppose that  
$h_i$, $i=1,\ldots,n$ are seminorms, which we can express in the
general form \smash{$h_i(v) =\max_{d \in D_i} \langle d,v \rangle$},
where $D_i \subseteq \R^{p_i}$ is a closed, convex set containing 0, 
for $i=1,\ldots,d$.  Suppose also that
\smash{$C_i=(X_i^T)^{-1}(D_i)=\{v \in \R^n: X_i^T v \in D_i\}$}, the 
inverse image of $D_i$ under the linear mapping $X_i^T$, for
$i=1,\ldots,d$. Then, perhaps surprisingly, it turns out that
the coordinate descent  
iterations \eqref{eq:cd} are exactly the same as the Dykstra
iterations \eqref{eq:dyk2}, via duality.  The key relationship is
extracted as a lemma below, for future reference, and then the formal 
equivalence is stated.  Proofs of these results, as with all
results in this paper, are given in the supplement.  

\begin{lemma}
\label{lem:dyk_cd_update}
Assume that $X_i \in \R^{n\times p_i}$ has full column rank and 
\smash{$h_i(v) =\max_{d \in D_i} \langle d,v \rangle$} for a
closed, convex set $D_i \subseteq \R^{p_i}$ containing 0.  Then for
\smash{$C_i=(X_i^T)^{-1}(D_i) \subseteq \R^n$} and any $b \in  
\R^n$,  
$$
\hat{w}_i = \argmin_{w_i \in \R^{p_i}} \; 
\half \|b - X_iw_i\|_2^2 + h_i(w_i) \iff
X_i \hat{w}_i = (\Id-P_{C_i})(b).
$$
where $\Id(\cdot)$ denotes the identity mapping.
\end{lemma}

\begin{theorem}
\label{thm:dyk_cd_equiv}
Assume the setup in Lemma \ref{lem:dyk_cd_update}, for each
$i=1,\ldots,d$. Then problems \eqref{eq:bap}, \eqref{eq:reg} are dual
to each other, and their solutions, denoted \smash{$\hat{u},\hat{w}$}, 
respectively, satisfy \smash{$\hat{u}=y-X\hat{w}$}.  Further,
Dykstra's algorithm \eqref{eq:dyk2} and coordinate descent
\eqref{eq:cd} are equivalent, and satisfy at all iterations
$k=1,2,3,\ldots$: 
$$
z_i^{(k)} = X_iw_i^{(k)} \quad\text{and}\quad
u_i^{(k)}=y - \sum_{j \leq i} X_j w_j^{(k)} - \sum_{j > i} X_j  
w_j^{(k-1)}, \quad\text{for $i=1,\ldots,d$}.
$$
\end{theorem}

The equivalence between coordinate descent and Dykstra's algorithm
dates back to (at least) \citet{han1988successive,gaffke1989cyclic},
under the special case $X=I$.  In fact, \citet{han1988successive}, 
presumably unaware of Dykstra's algorithm, seems to have reinvented
the method and established convergence through its relationship to
coordinate descent.  This work then inspired \citet{tseng1993dual}
(who must have also been unaware of Dykstra's algorithm) to improve
the existing analyses of coordinate descent, which at the time all
assumed smoothness of the objective function.  (Tseng continued on to
become 
arguably the single most important contributor to the theory of
coordinate descent of the 1990s and 2000s, and his seminal work
\citet{tseng2001convergence} is still one of the most comprehensive
analyses to date.) 

References to this equivalence can be found speckled throughout the
literature on Dykstra's method, but given the importance of
the regularized problem form \eqref{eq:reg} for modern statistical and 
machine learning estimation tasks, we feel that the connection between
Dykstra's algorithm and coordinate descent and is not well-known
enough and should be better explored.  In what follows, 
we show that some old work on Dykstra's algorithm, fed through this
equivalence, yields new convergence results for coordinate descent
for the lasso and a new parallel version of coordinate descent. 

\paragraph{ADMM.}
The augmented Lagrangian method of multipliers or ADMM was
invented by \citet{glowinski1975approximation,gabay1976dual}.  ADMM is 
a member of a class of methods generally called {\it operator
  splitting techniques}, and is equivalent (via a duality argument) to
{\it Douglas-Rachford splitting} 
\citep{douglas1956numerical,lions1979splitting}.  Recently, there has
been a strong revival of interest in ADMM (and operator splitting
techniques in general), arguably due (at least in part) to the popular 
monograph of \citet{boyd2011distributed}, where it is 
argued that the ADMM framework offers an appealing flexibility in
algorithm design, which permits parallelization in many nontrivial 
situations. As with coordinate descent, it would be difficult
thoroughly describe recent developments on ADMM, given the magnitude 
and pace of the literature on this topic. To give just a few examples,
recent progress includes finite-time linear convergence rates for ADMM 
(see \citealt{nishihara2015general,hong2017linear} and references
therein), and accelerated extensions of ADMM (see
\citealt{goldstein2014fast,kadkhodaie2015accelerated} and references
therein).  

To derive an ADMM algorithm for \eqref{eq:bap}, we introduce
auxiliary variables and equality constraints to put the problem in
a suitable ADMM form.  While different
formulations for the auxiliary variables and constraints give rise to   
different algorithms, loosely speaking, these algorithms
generally take on similar forms to Dykstra's algorithm for 
\eqref{eq:bap}.  The same is also true of ADMM for the {\it set
  intersection problem}, a simpler task than the best approximation
problem \eqref{eq:bap}, in which we only seek a point in the
intersection $C_1 \cap \cdots \cap C_d$, and solve  
\begin{equation}
\label{eq:sip}
\min_{u \in \R^n} \; \sum_{i=1}^d 1_{C_i}(u_i),
\end{equation}
where $1_C(\cdot)$ denotes the indicator function of a 
set $C$ (equal to 0 on $C$, and $\infty$ otherwise).  Consider
the case of $d=2$ sets, in which case the 
translation of \eqref{eq:sip} into ADMM form is unambiguous. ADMM for
\eqref{eq:sip}, properly initialized, appears highly similar to
Dykstra's algorithm for \eqref{eq:bap}; so similar, in fact, that
\citet{boyd2011distributed} mistook the two algorithms for being 
equivalent, which is not generally true, and was shortly thereafter
corrected by \citet{bauschke2013projection}.  

Below we show that when $d=2$, $C_1$ is a linear subspace, and $y \in
C_1$, an ADMM algorithm for \eqref{eq:bap} (and not the simpler set 
intersection problem \eqref{eq:sip}) is indeed equivalent to Dykstra's
algorithm for \eqref{eq:bap}. Introducing auxiliary
variables, the problem \eqref{eq:bap} becomes  
$$
\min_{u_1,u_2 \in \R^n} \; \|y-u_1\|_2^2 + 1_{C_1}(u_1) + 1_{C_2}(u_2)
\quad \st \quad u_1=u_2,
$$
and the augmented Lagrangian is 
\smash{$L(u_1,u_2,z) = \|y-u_1\|_2^2 + 1_{C_1}(u_1) + 1_{C_2}(u_2)
  + \rho\|u_1-u_2+z\|_2^2$}, where $\rho>0$ is an
augmented Lagrangian parameter. ADMM now repeats, for
$k=1,2,3,\ldots$:
\begin{equation}
\label{eq:admm}
\begin{aligned}
u_1^{(k)} &= P_{C_1}\bigg(\frac{y}{1+\rho} +  
\frac{\rho(u_2^{(k-1)}-z^{(k-1)})}{1+\rho}\bigg), \\
u_2^{(k)} &= P_{C_2}(u_1^{(k)}+z^{(k-1)}), \\
z^{(k)} &= z^{(k-1)} + u_1^{(k)}-u_2^{(k)}.
\end{aligned}
\end{equation}
Suppose we initialize \smash{$u_2^{(0)}=0$},
\smash{$z^{(0)}=0$}, and set $\rho=1$. Using linearity of
\smash{$P_{C_1}$}, the fact that $y \in C_1$, and a simple inductive
argument, the above iterations can be rewritten as 
\begin{equation}
\label{eq:admm2}
\begin{aligned}
u_1^{(k)} &= 
P_{C_1}(u_2^{(k-1)}), \\
u_2^{(k)} &= P_{C_2}(u_1^{(k)}+z^{(k-1)}), \\
z^{(k)} &= z^{(k-1)} + u_1^{(k)}-u_2^{(k)},
\end{aligned}
\end{equation}
which is precisely the same as Dykstra's iterations \eqref{eq:dyk2}, 
once we realize that, due again to linearity of \smash{$P_{C_1}$}, the
sequence \smash{$z_1^{(k)}$}, $k=1,2,3,\ldots$ in Dykstra's iterations
plays no role and can be ignored.  

Though $d=2$ sets in \eqref{eq:bap} may seem like a 
rather special case, the strategy for parallelization in both
Dykstra's algorithm and ADMM stems from rewriting a general $d$-set 
problem as a 2-set problem, so the above connection between Dykstra's
algorithm and ADMM can be relevant even for problems with $d>2$, and
will reappear in our later discussion of parallel coordinate
descent. As a matter of conceptual interest only, 
we note that for general $d$ (and no constraints on the sets being 
subspaces), Dykstra's iterations \eqref{eq:dyk2} can be viewed as a 
limiting version of the ADMM iterations either for \eqref{eq:bap} or
for \eqref{eq:sip}, as we send the augmented Lagrangian parameters  
to $\infty$ or to 0 at particular scalings.  See the supplement for
details. 
 
\section{Coordinate descent for the lasso}
\label{sec:lasso}

The {\it lasso} problem
\citep{tibshirani1996regression,chen1998atomic}, defined for a tuning
parameter $\lambda \geq 0$ as 
\begin{equation}
\label{eq:lasso}
\min_{w \in \R^p} \; \half\|y-Xw\|_2^2 + \lambda\|w\|_1, 
\end{equation}
is a special case of \eqref{eq:reg} where the coordinate blocks are of
each size 1, so that $X_i \in \R^n$, $i=1,\ldots,p$ are just
the columns of $X$, and $w_i \in \R$, $i=1,\ldots,p$ are the
components of  
$w$.  This problem fits into the framework of \eqref{eq:reg} with
\smash{$h_i(w_i)=\lambda|w_i|=\max_{d\in D_i} dw_i$} for 
$D_i=[-\lambda,\lambda]$, for each $i=1,\ldots,d$.

Coordinate descent is widely-used for the lasso 
\eqref{eq:lasso}, both because of the simplicity of the 
coordinatewise updates, and because careful implementations can 
achieve state-of-the-art performance, at the right problem sizes.
The use of coordinate descent for the lasso was popularized by 
\citet{friedman2007pathwise,friedman2010regularization}, but 
was studied earlier or concurrently by several others, e.g.,
\citet{fu1998penalized,sardy2000block,wu2008coordinate}. 

As we know from Theorem \ref{thm:dyk_cd_equiv}, the dual of
problem \eqref{eq:lasso} is the best approximation problem 
\eqref{eq:bap}, where 
\smash{$C_i=(X_i^T)^{-1}(D_i)=\{v \in \R^n: |X_i^T v|
\leq \lambda\}$} is an intersection of two halfspaces, for  
$i=1,\ldots,p$.  This makes
$C_1 \cap \cdots \cap C_d$ an intersection of
$2p$ halfspaces, i.e., a (centrally symmetric) polyhedron.  
For projecting onto a polyhedron, it is well-known that Dykstra's
algorithm reduces to {\it Hildreth's algorithm}
\citep{hildreth1957quadratic}, an older method for quadratic
programming that itself has an
interesting history in optimization.  Theorem 
\ref{thm:dyk_cd_equiv} hence shows coordinate descent for the
lasso \eqref{eq:lasso} is equivalent not only to Dykstra's 
algorithm, but also to Hildreth's algorithm, for \eqref{eq:bap}. 

This equivalence suggests a number of interesting directions to
consider. For example, key practical speedups have been developed for  
coordinate descent for the lasso that enable this method to attain 
state-of-the-art performance at the right problem sizes, such as
clever updating rules and screening rules
(e.g.,
\citealt{friedman2010regularization,elghaoui2012safe,tibshirani2012strong,  
wang2015dual}). These implementation tricks can now be used with
Dykstra's (Hildreth's) algorithm. On the flip side, as we show
next, older results from
\citet{iusem1990convergence,deutsch1994rate} on Dykstra's algorithm  
for polyhedra, lead to interesting new results 
on coordinate descent for the lasso.  

\begin{theorem}[\textbf{Adaptation of \citealt{iusem1990convergence}}]  
\label{thm:cd_lasso_rate_ip}
Assume the columns of $X \in \R^{n \times p}$ are in general position, 
and $\lambda>0$.
Then coordinate descent for the lasso \eqref{eq:lasso} has an
asymptotically linear convergence rate, in that for large
enough $k$, 
\begin{equation}
\label{eq:cd_lasso_rate_ip}
\frac{\|w^{(k+1)} - \hat{w}\|_\Sigma}{\|w^{(k)}-\hat{w}\|_\Sigma} \leq  
\bigg(\frac{a^2}{a^2 + \lambda_{\min}(X_A^T X_A)/\max_{i \in A} 
    \|X_i\|_2^2}\bigg)^{1/2},  
\end{equation}
where \smash{$\hat{w}$} is the lasso solution in
\eqref{eq:lasso}, $\Sigma=X^T X$, and
\smash{$\|z\|_\Sigma^2=z^T \Sigma z$} for $z \in \R^p$,
\smash{$A=\supp(\hat{w})$} is the active set of 
\smash{$\hat{w}$}, $a=|A|$ is its size, 
\smash{$X_A \in \R^{n \times a}$} denotes the columns of $X$ indexed  
by $A$, and \smash{$\lambda_{\min}(X_A^T X_A)$} denotes the smallest  
eigenvalue of \smash{$X_A^T X_A$}.
\end{theorem}

\begin{theorem}[\textbf{Adaptation of \citealt{deutsch1994rate}}]  
\label{thm:cd_lasso_rate_dh}
Assume the same conditions and notation as in Theorem 
\ref{thm:cd_lasso_rate_ip}.  Then for large enough $k$, 
\begin{equation}
\label{eq:cd_lasso_rate_dh}
\frac{\|w^{(k+1)} - \hat{w}\|_\Sigma}{\|w^{(k)}-\hat{w}\|_\Sigma} \leq  
\Bigg(1-\prod_{j=1}^{a-1} 
\frac{\|P_{\{i_{j+1},\ldots,i_a\}}^\perp
  X_{i_j}\|_2^2}{\|X_{i_j}\|_2^2}\Bigg)^{1/2},
\end{equation}
where we enumerate $A=\{i_1,\ldots,i_a\}$,
$i_1<\ldots<i_a$, and we denote by
\smash{$P_{\{i_{j+1},\ldots,i_a\}}^\perp$} the projection onto 
the orthocomplement of the column span of
\smash{$X_{\{i_{j+1},\ldots,i_a\}}$}.
\end{theorem}

The results in Theorems \ref{thm:cd_lasso_rate_ip},
\ref{thm:cd_lasso_rate_dh} both rely on the assumption 
of general position for the columns of $X$.  This is only used for
convenience and can be removed at the expense of more complicated
notation.  Loosely put, the general position condition simply rules
out trivial linear dependencies between small numbers of columns of
$X$, but places no restriction on the dimensions of $X$ (i.e., it
still allows for $p \gg n$).  It implies that the lasso
solution \smash{$\hat{w}$} is unique, and that $X_A$ (where
\smash{$A=\supp(\hat{w})$}) has full column rank. See
\citet{tibshirani2013lasso} 
for a precise definition of general position and proofs of these
facts.  We note that when $X_A$ has full column rank, the bounds in 
\eqref{eq:cd_lasso_rate_ip},  
\eqref{eq:cd_lasso_rate_dh} are strictly less than 1.

\begin{remark}[\textbf{Comparing \eqref{eq:cd_lasso_rate_ip} and
    \eqref{eq:cd_lasso_rate_dh}}] 
Clearly, both the bounds in \eqref{eq:cd_lasso_rate_ip},
\eqref{eq:cd_lasso_rate_dh} are adversely affected by correlations
among $X_i$, $i \in A$ (i.e., stronger correlations will bring
each closer to 1).  It seems to us that
\eqref{eq:cd_lasso_rate_dh} is usually 
the smaller of the two bounds, based on simple
mathematical and numerical comparisons. More detailed
comparisons would be interesting, but is beyond the scope of this
paper. 
\end{remark}

\begin{remark}[\textbf{Linear convergence without strong convexity}] 
One striking feature of the results in Theorems
\ref{thm:cd_lasso_rate_ip}, \ref{thm:cd_lasso_rate_dh} is that they
guarantee (asymptotically) linear convergence of the coordinate
descent iterates for the lasso, with no assumption about strong
convexity of the objective. More precisely, there are no restrictions
on the dimensionality of $X$, so we enjoy linear convergence {\it even  
without an assumption on the smooth part of the objective}.  This is
in line with classical results on coordinate descent for smooth
functions, see, e.g., \citet{luo1992convergence}.  The modern 
finite-time convergence analyses of coordinate descent do
not, as far as we understand, replicate this remarkable
property. 
For example, \citet{beck2013convergence,li2016improved} establish 
finite-time 
linear convergence rates for coordinate descent, but require strong
convexity of the entire objective.  
\end{remark}

\begin{remark}[\textbf{Active set identification}]
The asymptotics developed in
\citet{iusem1990convergence,deutsch1994rate} 
are based on a notion of (in)active set identification: the
critical value of $k$ after which \eqref{eq:cd_lasso_rate_ip}, 
\eqref{eq:cd_lasso_rate_dh} hold is based on the (provably finite)
iteration number at which Dykstra's algorithm identifies the inactive
halfspaces, i.e., at which coordinate descent identifies the inactive
set of variables, \smash{$A^c=\supp(\hat{w})^c$}.  This might help 
explain why in practice coordinate descent for the
lasso performs exceptionally well with warm starts, over a
decreasing sequence of tuning parameter values $\lambda$ (e.g.,
\citealt{friedman2007pathwise,friedman2010regularization}): here,
each coordinate descent run is likely to identify the
(in)active set---and hence enter the linear convergence phase---at an 
early iteration number. 
\end{remark}

\section{Parallel coordinate descent}
\label{sec:parallel}

\paragraph{Parallel-Dykstra-CD.} 
An important consequence of the connection between Dykstra's algorithm
and coordinate descent is a new parallel version of the latter,
stemming from an old parallel version of the former.  A parallel
version of Dykstra's algorithm is usually credited to
\citet{iusem1987simultaneous} for polyhedra and
\citet{gaffke1989cyclic} for general sets, but really the idea dates
back to the product space formalization of
\citet{pierra1984decomposition}.  We rewrite problem \eqref{eq:bap} as 
\begin{equation}
\label{eq:bap_prod}
\min_{u=(u_1,\ldots,u_d) \in \R^{nd}} \; \sum_{i=1}^d \gamma_i
\|y-u_i\|_2^2  
\quad \st \quad u \in C_0 \cap (C_1 \times \cdots \times C_d), 
\end{equation}
where \smash{$C_0=\{ (u_1,\ldots,u_d) \in \R^{nd} :
  u_1=\cdots=u_d\}$}, and $\gamma_1,\ldots,\gamma_d>0$ are  
weights that sum to 1.  
After rescaling appropriately to turn \eqref{eq:bap_prod}
into an unweighted best approximation problem, we can 
apply Dykstra's algorithm, which sets
\smash{$u_1^{(0)}=\cdots=u_d^{(0)}=y$},
\smash{$z_1^{(0)}=\cdots=z_d^{(0)}=0$}, and repeats:  
\begin{equation}
\begin{aligned}
\label{eq:dyk_par}
&u_0^{(k)} = \sum_{i=1}^d \gamma_i u_i^{(k-1)}, \\ 
&\begin{rcases*}
u_i^{(k)} = P_{C_i} (u_0^{(k)} + z_i^{(k-1)}), & \\
z_i^{(k)} = u_0^{(k)} + z_i^{(k-1)} - u_i^{(k)}, &
\end{rcases*}
\quad \text{for $i=1,\ldots,d$},
\end{aligned}
\end{equation}
for $k=1,2,3,\ldots$. The steps enclosed in 
curly brace above can all be performed in parallel, so that
\eqref{eq:dyk_par} is a parallel version of Dykstra's algorithm  
\eqref{eq:dyk2} for \eqref{eq:bap}.  
Applying Lemma \ref{lem:dyk_cd_update}, and a straightforward
inductive argument, the above algorithm can be rewritten as follows.
We set \smash{$w^{(0)}=0$}, and repeat:  
\begin{equation}
\label{eq:cd_dyk_par}
w_i^{(k)} = \argmin_{w_i \in \R^{p_i}} \;  
\half \Big\|y - Xw^{(k-1)} + X_i w_i^{(k-1)}/\gamma_i -  
X_i w_i/\gamma_i\Big\|_2^2 + h_i (w_i/\gamma_i),
\quad i=1,\ldots,d,
\end{equation}
for $k=1,2,3,\ldots$, which we call {\it parallel-Dykstra-CD} (with CD 
being short for coordinate descent).  
Again, note that each update in \eqref{eq:cd_dyk_par} can be performed
in parallel, so that \eqref{eq:cd_dyk_par} is a parallel version of 
coordinate descent \eqref{eq:cd} for \eqref{eq:reg}.  Also,
as \eqref{eq:cd_dyk_par} is just a reparametrization of
Dykstra's algorithm \eqref{eq:dyk_par} for the 2-set problem
\eqref{eq:bap_prod},  
it is guaranteed to converge in full generality,
as per the standard results on Dykstra's algorithm
\citep{han1988successive,gaffke1989cyclic}.  

\begin{theorem}
\label{thm:cd_dyk_par}
Assume that $X_i \in \R^{n\times p_i}$ has full column rank and  
\smash{$h_i(v) =\max_{d \in D_i} \langle d,v \rangle$} for a
closed, convex set $D_i \subseteq \R^{p_i}$ containing 0, for
$i=1,\ldots,d$. If \eqref{eq:reg} has a unique solution, then
the iterates in \eqref{eq:cd_dyk_par} converge to this solution. 
More generally, if the interior of \smash{$\cap_{i=1}^d
  (X_i^T)^{-1}(D_i)$}  
is nonempty, then the sequence \smash{$w^{(k)}$}, $k=1,2,3,\ldots$
from \eqref{eq:cd_dyk_par} has at least one accumulation point, and
any such point solves \eqref{eq:reg}.
Further, \smash{$Xw^{(k)}$}, $k=1,2,3,\ldots$ converges to 
\smash{$X\hat{w}$}, the optimal fitted value in \eqref{eq:reg}.   
\end{theorem}

There have been many recent exciting contributions to the parallel 
coordinate descent literature; two standouts
are \citet{jaggi2014communication,richtarik2016parallel}, 
and numerous others are described in \citet{wright2015coordinate}. 
What sets parallel-Dykstra-CD apart, perhaps, is its 
simplicity: convergence of the iterations \eqref{eq:cd_dyk_par}, given
in Theorem \ref{thm:cd_dyk_par}, just stems from the connection   
between coordinate descent and Dykstra's algorithm, and the fact that the  
parallel Dykstra iterations \eqref{eq:dyk_par} are nothing more than
the usual Dykstra iterations after a product space
reformulation.  Moreover, parallel-Dykstra-CD for the lasso enjoys an 
(asymptotic) linear convergence rate under essentially no assumptions,  
thanks once again to an old result on the parallel Dykstra (Hildreth) 
algorithm from \citet{iusem1990convergence}.   
The details can be found in the supplement. 

\paragraph{Parallel-ADMM-CD.}  As an alternative to the parallel
method derived using Dykstra's algorithm, ADMM can also offer a
version of parallel coordinate descent.  Since \eqref{eq:bap_prod} is
a best approximation problem with $d=2$ sets, we can refer back to our  
earlier ADMM algorithm in \eqref{eq:admm} for this problem.
By passing these ADMM iterations through the connection developed in
Lemma \ref{lem:dyk_cd_update}, we arrive at what we call {\it
  parallel-ADMM-CD}, which initializes \smash{$u_0^{(0)}=y$},
\smash{$w^{(-1)}=w^{(0)}=0$}, and repeats:
\begin{equation}
\begin{aligned}
\label{eq:cd_admm_par}
u_0^{(k)} &= \frac{(\sum_{i=1}^d \rho_i) u_0^{(k-1)}
}{1+\sum_{i=1}^d \rho_i} 
+ \frac{y-Xw^{(k-1)}}{1+\sum_{i=1}^d \rho_i}
+ \frac{X (w^{(k-2)}-w^{(k-1)})}{1+\sum_{i=1}^d \rho_i}, \\
w_i^{(k)} &= \argmin_{w_i \in \R^{p_i}} \;  
\half \Big\|u_0^{(k)} + X_i w_i^{(k-1)}/\rho_i -
X_i w_i/\rho_i\Big\|_2^2 + h_i (w_i/\rho_i),
\quad i=1,\ldots,d,
\end{aligned}
\end{equation}
for $k=1,2,3,\ldots$, where $\rho_1,\ldots,\rho_d>0$ are augmented  
Lagrangian parameters. In each iteration, the updates to
\smash{$w_i^{(k)}$}, $i=1,\ldots,d$ above can be done in 
parallel. Just based on their form, it seems that
\eqref{eq:cd_admm_par} can be seen as a parallel version of coordinate  
descent \eqref{eq:cd} for problem \eqref{eq:reg}. The next result
confirms this, leveraging standard theory for ADMM
\citep{gabay1983applications, eckstein1992douglas}.

\begin{theorem}
\label{thm:cd_admm_par}
Assume that $X_i \in \R^{n\times p_i}$ has full column rank and  
\smash{$h_i(v) =\max_{d \in D_i} \langle d,v \rangle$} for a
closed, convex set $D_i \subseteq \R^{p_i}$ containing 0, for
$i=1,\ldots,d$.  Then the sequence \smash{$w^{(k)}$}, $k=1,2,3,\ldots$ 
in \eqref{eq:cd_admm_par} converges to a solution in \eqref{eq:reg}. 
\end{theorem}

The parallel-ADMM-CD iterations in \eqref{eq:cd_admm_par} and 
parallel-Dykstra-CD iterations in \eqref{eq:cd_dyk_par} differ in 
that, where the latter uses a residual
\smash{$y-Xw^{(k-1)}$}, the former uses an iterate
\smash{$u_0^{(k)}$} that seems to have a more complicated form,  
being a convex combination of 
\smash{$u_0^{(k-1)}$} and \smash{$y-Xw^{(k-1)}$}, plus
a quantity that acts like a momentum term.  It turns out that
when $\rho_1,\ldots,\rho_d$ sum to 1, the two methods
\eqref{eq:cd_dyk_par}, \eqref{eq:cd_admm_par} are exactly the same.
While this may seem like a surprising coincidence, it is in fact
nothing more than a reincarnation of the previously established
equivalence between Dykstra's algorithm \eqref{eq:dyk2} and ADMM 
\eqref{eq:admm2} for a 2-set best approximation problem, 
as here $C_0$ is a linear subspace. 

Of course, with ADMM we need not choose probability
weights for $\rho_1,\ldots,\rho_d$, and the convergence in Theorem
\ref{thm:cd_admm_par} is guaranteed for any fixed values of
these parameters.  Thus, even though they were derived from
different perspectives, parallel-ADMM-CD subsumes parallel-Dykstra-CD, 
and it is a strictly more general approach. It is important to note that
larger values of $\rho_1,\ldots,\rho_d$ can often lead to
faster convergence, as we show in numerical experiments in the
supplement.  More detailed study and comparisons to related parallel  
methods are worthwhile, but are beyond the scope of this work. 

\section{Extensions and discussion}
\label{sec:extensions}

We studied connections between Dykstra's algorithm, ADMM, and
coordinate descent, leveraging these connections to establish an  
(asymptotically) linear convergence rate for coordinate descent for
the lasso, as well as two parallel versions of coordinate
descent (one based on Dykstra's algorithm and the other on ADMM). 
Some extensions and possibilities for future work are described
below. 

\paragraph{Nonquadratic loss: Dykstra's algorithm and 
coordinate descent.}  Given a convex function $f$, a  
generalization of \eqref{eq:reg} is the {\it regularized 
estimation problem}
\begin{equation}
\label{eq:reg_gen}
\min_{w \in \R^p} \; f(Xw) + \sum_{i=1}^d h_i(w_i).
\end{equation}
Regularized regression \eqref{eq:reg} is given by
\smash{$f(z)=\half\|y-z\|_2^2$}, and e.g., regularized classification 
(under the logistic loss) by
\smash{$f(z)=-y^Tz+\sum_{i=1}^n\log(1+e^{z_i})$}.  
In (block) coordinate descent for \eqref{eq:reg_gen}, we initialize
say \smash{$w^{(0)}=0$}, and repeat, for $k=1,2,3,\ldots$: 
\begin{equation}
\label{eq:cd_gen}
w_i^{(k)} = \argmin_{w_i \in \R^{p_i}} \; 
f\bigg(\sum_{j < i} X_j w_j^{(k)} + \sum_{j > i} X_jw_j^{(k-1)} + 
X_iw_i\bigg) + h_i (w_i), 
\quad i=1,\ldots,d. 
\end{equation}
On the other hand, given a differentiable and strictly convex function
$g$, we can generalize \eqref{eq:bap} to the following {\it best
  Bregman-approximation problem},
\begin{equation}
\label{eq:bap_gen}
\min_{u \in \R^n} \; D_g(u,b)  
\quad \st \quad u \in C_1 \cap \cdots \cap C_d. 
\end{equation}
where \smash{$D_g(u,b)=g(u) - g(b) - \langle \nabla g(b), u-b\rangle$}
is the {\it Bregman divergence} between $u$ and $b$ with respect to
$g$.  When 
\smash{$g(v)=\half\|v\|_2^2$} (and $b=y$), this recovers the
best approximation problem \eqref{eq:bap}.  As shown in
\citet{censor1998dykstra,bauschke2000dykstra}, Dykstra's algorithm can 
be extended to apply to \eqref{eq:bap_gen}.  We initialize
\smash{$u_d^{(0)}=b$}, \smash{$z_1^{(0)}=\cdots=z_d^{(0)}=0$}, and
repeat for $k=1,2,3,\ldots$:
\begin{equation}
\begin{aligned}
\label{eq:dyk_gen}
&u_0^{(k)} = u_d^{(k-1)}, \\
&\begin{rcases*}
u_i^{(k)} = (P^g_{C_i} \circ \nabla g^*)\Big(\nabla g (u_{i-1}^{(k)})
+ z_i^{(k-1)}\Big), & \\ 
z_i^{(k)} = \nabla g(u_{i-1}^{(k)}) + z_i^{(k-1)} - 
\nabla g(u_i^{(k)}), & 
\end{rcases*}
\quad \text{for $i=1,\ldots,d$},
\end{aligned}
\end{equation}
where \smash{$P_C^g(x)=\argmin_{c \in C} D_g(c,x)$} denotes the
Bregman (rather than Euclidean) projection of $x$ onto a set $C$, and
$g^*$ is the conjugate function of $g$. Though it may not be
immediately obvious, when \smash{$g(v)=\half\|v\|_2^2$} 
the above iterations \eqref{eq:dyk_gen} reduce to the standard
(Euclidean) Dykstra iterations in \eqref{eq:dyk2}.
Furthermore, Dykstra's algorithm and coordinate descent are  
equivalent in the more general setting. 

\begin{theorem}
\label{thm:dyk_cd_equiv_gen}
Let $f$ be a closed, strictly convex, differentiable function.
Assume that $X_i \in \R^{n\times p_i}$ has full
column rank, and \smash{$h_i(v) =\max_{d \in D_i} \langle d,v
  \rangle$} for a closed, convex set $D_i \subseteq \R^{p_i}$
containing 0, for
$i=1,\ldots,d$.  Also, let $g(v)=f^*(-v)$, $b=-\nabla f(0)$, and 
\smash{$C_i=(X_i^T)^{-1}(D_i) \subseteq \R^n$}, $i=1,\ldots,d$. Then
problems  
\eqref{eq:reg_gen},  \eqref{eq:bap_gen} are dual to each other,
and their solutions \smash{$\hat{w},\hat{u}$} satisfy
\smash{$\hat{u}=-\nabla f(X\hat{w})$}. Further, Dykstra's
algorithm \eqref{eq:dyk_gen} and coordinate descent
\eqref{eq:cd_gen} are equivalent, i.e., for $k=1,2,3,\ldots$: 
$$
z_i^{(k)} = X_iw_i^{(k)} \quad\text{and}\quad
u_i^{(k)}=-\nabla f\bigg(\sum_{j \leq i} X_j w_j^{(k)} + \sum_{j > i}  
X_j w_j^{(k-1)}\bigg), \quad\text{for $i=1,\ldots,d$}. 
$$
\end{theorem}

\paragraph{Nonquadratic loss: parallel coordinate descent methods.} 
For a general regularized estimation problem \eqref{eq:reg_gen},
parallel coordinate descent methods can be derived by applying
Dykstra's algorithm and ADMM to a product space reformulation of the
dual.  Interestingly, the subsequent coordinate descent algorithms are
{\it no longer equivalent} (for a unity augmented Lagrangian
parameter), and they feature complementary computational
structures. The Dykstra version has a closed-form $u_0$-update, but
its (parallel) $w$-updates require coordinatewise minimizations
involving the smooth, convex loss $f$.  On the other hand, the ADMM
version admits a more difficult $u_0$-update, but its (parallel)
$w$-updates only require coordinatewise minimizations with a quadratic
loss (this being typically simpler than the corresponding
minimizations for most nonquadratic $f$ of interest).  The supplement
gives details.

\paragraph{Asynchronous parallel algorithms,
  and coordinate descent in Hilbert spaces.}  We finish with some
directions for possible future work.  Asynchronous variants of
parallel coordinate descent are currently of great interest, e.g., see
the review in \citet{wright2015coordinate}.  Given the
link between ADMM and coordinate descent developed in this
paper, it would be interesting to investigate the implications of the 
recent exciting progress on asynchronous ADMM, e.g., see 
\citet{chang2016asynchronousI,chang2016asynchronousII} and references
therein, for coordinate descent.  In a separate direction, much of the
literature on Dykstra's algorithm emphasizes that this method works
seamlessly in Hilbert spaces.  It would be interesting to consider the
connections to (parallel) coordinate descent in infinite-dimensional
function spaces, which we would encounter, e.g., in alternating
conditional expectation algorithms or backfitting algorithms in
additive models. 

\newpage
\begin{center}
{\LARGE Appendix: Proofs, Technical Details, and Numerical
  Experiments} 
\end{center}

\renewcommand\thesection{A.\arabic{section}}   
\renewcommand\thefigure{A.\arabic{figure}}    
\renewcommand\theequation{A.\arabic{equation}}     
\renewcommand\thetheorem{A.\arabic{theorem}}     
\renewcommand\thelemma{A.\arabic{lemma}}     
\renewcommand\theremark{A.\arabic{remark}}  
\setcounter{section}{0}   
\setcounter{figure}{0}
\setcounter{equation}{0}
\setcounter{theorem}{0}
\setcounter{lemma}{0}
\setcounter{remark}{0}

\section{Proofs of Lemma \ref{lem:dyk_cd_update} and Theorem
    \ref{thm:dyk_cd_equiv}}

These results are direct consequences of the more general Lemma 
\ref{lem:dyk_cd_update_gen} and Theorem \ref{thm:dyk_cd_equiv_gen},
respectively, when \smash{$f(v)=\half\|y-v\|_2^2$} (and so
\smash{$f^*(v)=-\half\|y\|_2^2+\half\|y+v\|_2^2$}); see Section  
\ref{sec:dyk_cd_gen} below for their proofs.   

\section{Dykstra's algorithm and ADMM for the $d$-set best
  approximation and set intersection problems}

Here we show that, under an inertial-type modification, the ADMM 
iterations for \eqref{eq:sip} are in a certain limiting sense
equivalent to Dykstra's iterations for \eqref{eq:bap}.  We can
introducing auxiliary variables to transform problem \eqref{eq:bap}
into   
$$
\min_{u_0,\ldots,u_d \in \R^n} \; \sum_{i=1}^d
1_{C_i}(u) \quad \st \quad u_d=u_0, u_0=u_1, \ldots, u_{d-1}=u_d,
$$
and the corresponding augmented Lagrangian is 
\smash{$L(u_0,\ldots,u_d,z_0,\ldots,z_d) =
\rho_0 \|u_d-u_0+z_0\|_2^2+{}$} \smash{$\sum_{i=1}^d  
(1_{C_i}(u)+ \rho_i \|u_{i-1}-u_i+z_i\|_2^2)$},
where $\rho_0,\ldots,\rho_d>0$ are augmented Lagrangian parameters.
ADMM is defined by repeating the updates:
\begin{align*}
u_i^{(k)} &= \argmin_{u_i \in \R^n} \;
L(u_0^{(k)},\ldots,u_{i-1}^{(k)},u_i,u_{i+1}^{(k-1)},\ldots,u_d^{(k-1)}),
\quad i=0,\ldots,d, \\
z_i^{(k)} &= z_i^{(k-1)} + u_{i-1}^{(k)} - u_i^{(k)}, 
\quad i=0,\ldots,d, 
\end{align*}
for $k=1,2,3,\ldots$,
where we use \smash{$u_{-1}^{(k)}=u_d^{(k)}$} for convenience.  Now
consider an inertial modification in which, for the $u_0$ update
above, we add the term \smash{$\|u_0-u_d^{(k-1)}\|_2^2$} to the
augmented Lagrangian in the 
minimization.  A straightforward derivation then leads to the
ADMM updates:
\begin{equation}
\begin{aligned}
\label{eq:admm3}
u_0^{(k)} &= \frac{u_d^{(k-1)}}{1+\rho_0+\rho_1} + 
\frac{\rho_0(u_d^{(k-1)}+z_0^{(k-1)})}
{1+\rho_0+\rho_1} + \frac{\rho_1(u_1^{(k-1)}-z_1^{(k-1)})} 
{1+\rho_0+\rho_1}, \\
u_i^{(k)} &= P_{C_i} \bigg(
\frac{(u_{i-1}^{(k)}+z_i^{(k-1)})}{1+\rho_{i+1}/\rho_i} +
\frac{(\rho_{i+1}/\rho_i)(u_{i+1}^{(k-1)}-z_{i+1}^{(k-1)})}{1+\rho_{i+1}/\rho_i}   
\bigg), \quad i=1,\ldots,d-1, \\
u_d^{(k)} &= P_{C_d} \bigg(
\frac{(u_{d-1}^{(k)}+z_d^{(k-1)})}{1+\rho_0/\rho_d} +
\frac{(\rho_0/\rho_d)(u_0^{(k)}-z_0^{(k-1)})}{1+\rho_0/\rho_d} \bigg),
\\ z_i^{(k)} &= z_i^{(k-1)} + u_{i-1}^{(k)} - u_i^{(k)},  
\quad i=0,\ldots,d,  
\end{aligned}
\end{equation}
for $k=1,2,3,\ldots$. 
Under the choices \smash{$\rho_0 = \alpha^{d+1}$} and 
\smash{$\rho_i=\alpha^i$}, $i=1,\ldots,d$, we see that as 
$\alpha \to 0$ the ADMM iterations \eqref{eq:admm3} exactly coincide 
with the Dykstra iterations \eqref{eq:dyk2}.  Thus, under the proper
initializations, 
\smash{$u_d^{(0)}=y$} and
\smash{$z_0^{(0)}=\cdots=z_d^{(0)}=0$},
the limiting ADMM algorithm for \eqref{eq:sip} matches Dykstra's
algorithm for \eqref{eq:bap}.  

Similar arguments can be used equate ADMM for \eqref{eq:bap} to
Dykstra's algorithm, again in limiting
sense.  We rewrite \eqref{eq:bap} as
$$
\min_{u_0,\ldots,u_d \in \R^n} \; \|y-u_0\|_2^2 + \sum_{i=1}^d
1_{C_i}(u) \quad \st \quad u_d=u_0, u_0=u_1, \ldots, u_{d-1}=u_d.
$$
Using an inertial modification for the $u_0$ update, where we now add
the term \smash{$\rho_{-1}\|u_0-u_d^{(k-1)}\|_2^2$} to the augmented
Lagrangian in the minimization, the ADMM updates become:
\begin{equation}
\begin{aligned}
\label{eq:admm4}
u_0^{(k)} &= \frac{y}{1+\rho_{-1} + \rho_0+\rho_1} +
\frac{\rho_{-1} u_d^{(k-1)}}{1+\rho_{-1} + \rho_0+\rho_1} + 
\frac{\rho_0(u_d^{(k-1)}+z_0^{(k-1)})}
{1+\rho_{-1} + \rho_0+\rho_1}+ \frac{\rho_1(u_1^{(k-1)}-z_1^{(k-1)})} 
{1+\rho_{-1} + \rho_0+\rho_1}, \\
u_i^{(k)} &= P_{C_i} \bigg(
\frac{(u_{i-1}^{(k)}+z_i^{(k-1)})}{1+\rho_{i+1}/\rho_i} +
\frac{(\rho_{i+1}/\rho_i)(u_{i+1}^{(k-1)}-z_{i+1}^{(k-1)})}{1+\rho_{i+1}/\rho_i}   
\bigg), \quad i=1,\ldots,d-1, \\
u_d^{(k)} &= P_{C_d} \bigg(
\frac{(u_{d-1}^{(k)}+z_d^{(k-1)})}{1+\rho_0/\rho_d} +
\frac{(\rho_0/\rho_d)(u_0^{(k)}-z_0^{(k-1)})}{1+\rho_0/\rho_d} \bigg),
\\ z_i^{(k)} &= z_i^{(k-1)} + u_{i-1}^{(k)} - u_i^{(k)},  
\quad i=0,\ldots,d,  
\end{aligned}
\end{equation}
for $k=1,2,3,\ldots$.  Setting \smash{$\rho_{-1}=\alpha^{d+1}$},
\smash{$\rho_0=1$}, and \smash{$\rho_i=\alpha^{d+1-i}$},
$i=1,\ldots,d$, we can see that as $\alpha \to \infty$, the ADMM
iterations \eqref{eq:admm4} converge to the Dykstra iterations 
\eqref{eq:dyk2}, and therefore with initializations
\smash{$u_d^{(0)}=y$} and
\smash{$z_0^{(0)}=\cdots=z_d^{(0)}=0$},
the limiting ADMM algorithm for \eqref{eq:bap} coincides with
Dykstra's algorithm for the same problem.

The links above between ADMM and Dykstra's algorithm are
intended to be of conceptual interest, and the ADMM algorithms  
\eqref{eq:admm3},  
\eqref{eq:admm4} may not be practically useful for arbitrary
configurations of the augmented Lagrangian parameters.  After all,
both of these are multi-block ADMM approaches, and multi-block ADMM
has subtle convergence behavior as studied,
e.g., in \citet{lin2015global,chen2016direct}.

\section{Proof of Theorem \ref{thm:cd_lasso_rate_ip}}
\label{sec:cd_lasso_rate_ip}

By Theorem \ref{thm:dyk_cd_equiv}, we know that coordinate
descent applied to the lasso problem \eqref{eq:lasso} is equivalent to
Dykstra's algorithm on the best approximation problem \eqref{eq:bap},
with \smash{$C_i=\{v \in \R^n: |X_i^T
  v| \leq \lambda\}$}, for $i=1,\ldots,p$.  In particular, at
the end of the $k$th iteration, it holds that
$$
u_p^{(k)} = y - Xw^{(k)}, \quad \text{for $k=1,2,3,\ldots$}. 
$$
By duality, we also have \smash{$\hat{u}=y-X\hat{w}$} at the solutions
\smash{$\hat{u},\hat{w}$} in \eqref{eq:bap}, \eqref{eq:lasso}, 
respectively.  Therefore any statement about the convergence of  
Dykstra's iterates may be translated into a statement about the
convergence of the coordinate descent iterates, via the relationship 
\begin{equation}
\label{eq:dyk_cd_relationship}
\|u^{(k)}-\hat{u}\|_2 = \|Xw^{(k)}-X\hat{w}\|_2 =
\|w^{(k)}-\hat{w}\|_\Sigma,
\quad \text{for $k=1,2,3,\ldots$}.
\end{equation}
We seek to apply the main result from
\citet{iusem1990convergence}, on   
the asymptotic convergence rate of Dykstra's (Hildreth's) algorithm  
for projecting onto a polyhedron.  One slight complication is that,
in the current paramterization, coordinate descent is equivalent to
Dykstra's algorithm on 
$$
C_1 \cap \ldots \cap C_p = \bigcap_{i=1}^p \{v \in \R^n : |X_i^T v|
\leq \lambda\}, 
$$
While polyhedral, the above is not explicitly an intersection of 
halfspaces (it is an intersection of slabs), which is the setup
required by the analysis of \citet{iusem1990convergence}.
Of course, we can simply define
\smash{$C_i^+=\{v \in \R^n: X_i^T v \leq \lambda\}$} and
\smash{$C_i^-=\{v \in \R^n: X_i^T v \geq -\lambda\}$}, $i=1,\ldots,p$,
and then the above intersection is equivalent to 
$$
C_1^+ \cap C_1^- \cap \ldots \cap C_p^+ \cap C_p^- = 
\bigcap_{i=1}^p \Big( \{v \in \R^n : X_i^T v \leq \lambda\} \cap  
\{v \in \R^n : X_i^T v \geq -\lambda\} \Big).
$$
Moreover, one can check that the iterates from Dykstra's algorithm on  
\smash{$C_1^+ \cap C_1^- \cap \ldots \cap C_p^+ \cap C_p^-$} 
match\footnote{By this we mean that
\smash{$u_i^{-,(k)}=u_i^{(k)}$} for all $i=1,\ldots,p$ and
$k=1,2,3,\ldots$, if the iterates from Dykstra's algorithm on 
\smash{$C_1^+ \cap C_1^- \cap \ldots \cap C_p^+ \cap C_p^-$} are
denoted as \smash{$u^{+,(k)}_i, u^{-,(k)}_i$}, $i=1,\ldots,p$.} those 
from Dykstra's algorithm on \smash{$C_1 \cap \ldots \cap C_p$},
provided that the algorithms cycle over the sets in the order
they are written in these intersections.  This means that the
analysis of \cite{iusem1990convergence} can be applied to coordinate
descent for the lasso. 

The error constant from Theorem 1 in
\citet{iusem1990convergence} is based on a geometric quantity that we 
explicitly lower bound below.  It is not clear to us
whether our lower bound is the best possible, and a better lower
bound would improve the error constant presented in Theorem 
\ref{thm:cd_lasso_rate_ip}.  

\begin{lemma}
\label{lem:mu_lower_bound}
Let $H_i=\{x \in \R^n : h_i^T x = b_i\}$, $i=1,\ldots,s$ be hyperplanes,
and $S=H_1 \cap \ldots \cap H_s$ the $s$-dimensional affine subspace
formed by their intersection.  For each $x \in \R^n$, denote by $H_x$
the hyperplane among $H_1,\ldots,H_s$ farthest from $x$.  Define 
$$
\mu = \inf_{x \in \R^n} \; \frac{d(x,H_x)}{d(x,S)}, 
$$
where \smash{$d(x,S) = \inf_{y \in S} \|x-y\|_2$} is the 
distance between $x$ and $S$, and similarly for $d(x,H_x)$.  Then 
$$
\mu \geq \frac{\sigma_{\min}(M)}{\sqrt{s} 
\max_{i=1,\ldots,s} \|h_i\|_2} > 0, 
$$
where $M \in \R^{n \times s}$ has columns
$h_1,\ldots,h_s$, and \smash{$\sigma_{\min}(M)$} is its smallest
nonzero singular value.
\end{lemma}
\begin{proof}
For any $x \in \R^n$, note that \smash{$d(x,S)=\|M^+(b-M^T x)\|_2$},
where $M^+$ is the Moore-Penrose pseudoinverse of $M$.  Also,  
\smash{$d(x,H_x)=\max_{i=1,\ldots,s} |b_i-h_i^T x|/\|h_i\|_2$}. Hence,
writing \smash{$\sigma_{\max}(M^+)$} for the maximum singular value of 
$M^+$, 
\begin{align*}
\frac{d(x,H_x)}{d(x,S)} &\geq 
\frac{\max_{i=1,\ldots,s} |b_i-h_i^T x|/\|h_i\|_2}
{\sigma_{\max}(M^+) \|b-M^T x\|_2} \\
&\geq \frac{\sigma_{\min}(M)}{\max_{i=1,\ldots,s} \|h_i\|_2}  
\frac{\max_{i=1,\ldots,s} |b_i-h_i^T x|}{\|b-M^T x\|_2} \\
&\geq \frac{\sigma_{\min}(M)}{\sqrt{s} \max_{i=1,\ldots,s} \|h_i\|_2}, 
\end{align*}
where we have used the fact that
\smash{$\sigma_{\max}(M^+)=1/\sigma_{\min}(M)$}, as well as 
\smash{$\|v\|_\infty/\|v\|_2 \geq 1/\sqrt{s}$} for all vectors $v \in 
\R^s$.  Taking an infimum over $x \in \R^n$ establishes the
result. 
\end{proof}

Now we adapt and refine the result in Theorem 1 from
\citet{iusem1990convergence}.  These authors show that for large
enough $k$, 
$$
\frac{\|u_p^{(k+1)} - \hat{u}\|_2}{\|u_p^{(k)} - \hat{u}\|_2} \leq 
\bigg(\frac{1}{1+\sigma}\bigg)^{1/2},
$$
where $\sigma = \mu^2/p$, and $\mu>0$ is defined as follows.  
Let \smash{$A=\{ i \in \{1,\ldots,p\} : |X_i^T 
  \hat{u}|=\lambda\}$}, and let \smash{$\rho=\sign(X_A^T \hat{u})$}.   
Also let
$$
H_i=\{v \in \R^n : X_i^T v = \rho_i \lambda\}, \quad
i \in A,
$$
as well as \smash{$S=\cap_{i \in A} H_i$}.  Then  
$$
\mu = \inf_{x \in \R^n} \; \frac{d(x,H_x)}{d(x,S)}, 
$$
where for each $x \in \R^n$, we denote by $H_x$ the hyperplane among
$H_i$, $i \in A$ farthest from $x$.  

In the nomenclature of the lasso
problem, the set $A$ here is known as the {\it
  equicorrelation set}.  The general position assumption on $X$
implies that the lasso \smash{$\hat{w}$} solution is unique, and that
(for almost every in $y \in \R^n$), the equicorrelation set and
support of \smash{$\hat{w}$} are equal, so we can write
\smash{$A=\supp(\hat{w})$}.  See \citet{tibshirani2013lasso}.   

From Lemma \ref{lem:mu_lower_bound}, we have that \smash{$\mu^2 \geq
  \lambda_{\min}(X_A^T X_A)/(a \max_{i \in A} \|X_i\|^2_2)$}, where
$a=|A|$, and so   
$$
\bigg(\frac{1}{1+\sigma}\bigg)^{1/2} \leq 
\bigg(\frac{pa}{pa+\lambda_{\min}(X_A^T X_A)/
\max_{i \in A} \|X_i\|^2_2}\bigg)^{1/2}.
$$
This is almost the desired result, but it is weaker, because of its
dependence on $pa$ rather than $a^2$.  Careful inspection of the
proof of Theorem 1 in \citet{iusem1990convergence} shows that the
factor of $p$ in the constant $\sigma=\mu^2/p$ comes from an
application of Cauchy-Schwartz, to derive an upper bound of
the form (translated to our notation):
$$
\bigg(\sum_{i=1}^{p-1} \|u_{i+1}^{(k)}-u_i^{(k)}\|_2 \bigg)^2 \leq 
p \sum_{i=1}^{p-1} \|u_{i+1}^{(k)}-u_i^{(k)}\|_2^2.
$$
See their equation (33) (in which, we note, there is a typo: the
entire summation should be squared).  However, in
the summation on the left above, at most $a$ of the above terms are
zero. This is true as 
\smash{$u_{i+1}^{(k)}-u_i^{(k)}=X_{i+1}w_{i+1}^{(k)}-X_{i+1}w_{i+1}^{(k-1)}$},
$i=1,\ldots,p-1$, and for large enough values of $k$, as considered by 
these authors, we will have \smash{$w_i^{(k)}=0$} for all $i \notin
A$, as shown in Lemma 1 by \citet{iusem1990convergence}.  Thus the
last display can be sharpened to 
$$
\bigg(\sum_{i=1}^{p-1} \|u_{i+1}^{(k)}-u_i^{(k)}\|_2 \bigg)^2 \leq 
a \sum_{i=1}^{p-1} \|u_{i+1}^{(k)}-u_i^{(k)}\|_2^2,
$$
which allows to define $\sigma=\mu^2/a$.  Retracing through the steps
above to upper bound \smash{$(1/1+\sigma)^{1/2}$}, and applying
\eqref{eq:dyk_cd_relationship}, then leads to the result as stated in
the theorem. 

\section{Proof of Theorem \ref{thm:cd_lasso_rate_dh}}

As in the proof of Theorem \ref{thm:cd_lasso_rate_ip}, we observe that  
the relationship \eqref{eq:dyk_cd_relationship} between the 
Dykstra iterates and coordinate descent iterates allows us to turn a  
statement about the convergence of the latter into one about
convergence of the former. We consider Theorem 3.8 in
\citet{deutsch1994rate}, on the asymptotically linear convergence rate
of Dykstra's (Hildreth's) algorithm for projecting onto an
intersection of halfspaces (we note here, as explained in the proof of
Theorem \ref{thm:cd_lasso_rate_ip}, that coordinate descent for the
lasso can be equated to Dykstra's algorithm on halfspaces, even though
in the original dual formulation, Dykstra's algorithm operates on
slabs). 

Though the error constant is not explicitly written in the statement
of Theorem 3.8 in \citet{deutsch1994rate}\footnote{The result in
  Theorem 3.8 of \citet{deutsch1994rate} is actually written in
  nonasymptotic form, 
  i.e., it is stated (translated to our notation) that
\smash{$\|u_d^{(k)}-\hat{u}\|_2 \leq \rho c^k$}, for some constants
$\rho>0$ and $0<c<1$, and all iterations $k=1,2,3,\ldots$. The error
constant $c$ can be explicitly characterized, as we show in the proof
of Theorem 
\ref{thm:cd_lasso_rate_dh}. But the constant $\rho$ cannot be, and in
fact, the nonasymptotic error bound in \citet{deutsch1994rate} is
really nothing more than a restatement of the more precise asymptotic
error bound developed in the proofs of their Lemma 3.7 and Theorem
3.8. Loosely put, any asymptotic error bound can be transformed into a  
nonasymptotic one by simply defining a problem-specific constant
$\rho$ to be large enough that it makes the bound valid until the
asymptotics kick in. This describes the strategy taken in
\citet{deutsch1994rate}.}, the proofs of Lemma 3.7 and Theorem 3.8
from these authors reveals the following.  Define \smash{$A=\{ i \in
  \{1,\ldots,p\} : |X_i^T \hat{u}|=\lambda\}$}, and enumerate
$A=\{i_1,\ldots,i_a\}$ with $i_1<\ldots<i_a$.  
As in the proof of Theorem \ref{thm:cd_lasso_rate_ip}, we note that
the general position assumption on $X$ allows us to write (almost
everywhere in $y \in \R^n$) \smash{$A=\supp(\hat{w})$}, for the unique
lasso solution \smash{$\hat{w}$}. Also define 
$$
H_{i_j}=\{v \in \R^n : X_{i_j}^T v = 0\}, \quad \text{for
  $j=1,\ldots,a$}. 
$$
\citet{deutsch1994rate} show that, for large enough $k$,
\begin{equation}
\label{eq:asymp_linear_dh}
\frac{\|u_p^{(k+1)} - \hat{u}\|_2}{\|u_p^{(k)} - \hat{u}\|_2} \leq 
\max_{\substack{B \subseteq A, \\
B=\{\ell_1,\ldots,\ell_b\}, \\
\ell_1 < \ldots < \ell_b}} \Bigg(
1-\prod_{j=1}^{b-1} 
\Big(1-c^2\big( H_{\ell_j}, H_{\ell_{j+1}} \cap \cdots \cap
H_{\ell_b}\big) \Big) \Bigg), 
\end{equation}
where $c(L,M)$ denotes the cosine of the angle between linear
subspaces $L,M$.  Now, to simplify the bound on the right-hand side
above, we make two observations.  First, we observe that in
general \smash{$c(L,M)=c(L^\perp,M^\perp)$} (as in, e.g., Theorem
3.5 of \citealt{deutsch1994rate}), so we have
\begin{align*}
c\Big( H_{\ell_j}, H_{\ell_{j+1}} \cap \cdots \cap H_{\ell_b} \Big) &=
c\Big( H_{\ell_j}^\perp, (H_{\ell_{j+1}} \cap \cdots \cap
H_{\ell_b})^\perp \Big) \\ &= 
c\Big(\col(X_{\ell_j}),\col(X_{\{\ell_{j+1},\ldots,\ell_b\}})\Big) \\
&= \frac{\|P_{\{\ell_{j+1},\ldots,\ell_b\}} X_{\ell_j}\|_2}
{\|X_{\ell_j}\|_2}, 
\end{align*}
where in the last line we used that the cosine of the angle between
subspaces has an explicit form, when one of these subspaces is
1-dimensional.  Second, we observe that the maximum in
\eqref{eq:asymp_linear_dh} is actually achieved at $B=A$, since the
cosine of the angle between a 1-dimensional subspace and a second
subspace can only increase when the second subspace is made larger. 
Putting these two facts together, and using
\eqref{eq:dyk_cd_relationship}, establishes the result in the
theorem. 

\section{Derivation details for \eqref{eq:dyk_par},
  \eqref{eq:cd_dyk_par}
and proof of Theorem \ref{thm:cd_dyk_par}} 
\label{sec:cd_dyk_par}

By rescaling, problem \eqref{eq:bap_prod} can be written as
\begin{equation}
\label{eq:bap_prod2}
\min_{\tilde{u} \in \R^{nd}} \;
\|\tilde{y}-\tilde{u}\|_2^2 
\quad \st \quad \tilde{u} \in \tilde{C}_0 \cap (\tilde{C}_1 \times
\cdots \times \tilde{C}_d),  
\end{equation}
where \smash{$\tilde{y}=(\sqrt\gamma_1 y, \ldots, \sqrt\gamma_d y)
  \in \R^{nd}$}, and
$$
\tilde{C}_0 = \{ (v_1,\ldots,v_d) \in \R^{nd} :
v_1/\sqrt\gamma_1 = \cdots = v_d/\sqrt\gamma_d\}
\quad\text{and}\quad 
\tilde{C}_i = \sqrt\gamma_i C_i, \quad \text{for $i=1,\ldots,d$}. 
$$
The iterations in \eqref{eq:dyk_par} then follow by applying
Dykstra's algorithm to \eqref{eq:bap_prod2}, transforming the iterates
back to the original scale (so that the projections are all in terms
of $C_0,C_1,\ldots,C_d$), and recognizing
that the sequence say \smash{$z_0^{(k)}$}, $k=1,2,3,\ldots$ that
would usually accompany \smash{$u_0^{(k)}$}, $k=1,2,3,\ldots$ is not 
needed because $C_0$ is a linear subspace. 

As for the representation \eqref{eq:cd_dyk_par}, it can be verified
via a simple inductive argument that the Dykstra iterates in
\eqref{eq:dyk_par} satisfy, for all $k=1,2,3,\ldots$, 
$$
u_0^{(k)} = y - \sum_{i=1}^d \gamma_i z_i^{(k-1)}, 
\quad i=1,\ldots,d.
$$
Also, as shown in the proof of Theorem \ref{thm:dyk_cd_equiv_gen} in  
Section \ref{sec:dyk_cd_gen} below,
the image of the residual projection operator \smash{$\Id-P_{C_i}$} is
contained in the column span of $X_i$, for each $i=1,\ldots,d$.  This
means that we can parametrize the Dykstra iterates, for
$k=1,2,3,\ldots$, as 
$$
u_0^{(k)} = y - \sum_{i=1}^d \gamma_i X_i\tilde{w}_i^{(k-1)} 
\quad \text{and} \quad   
z^{(k)}_i = X_i\tilde{w}_i^{(k)}, \quad i=1,\ldots,d, 
$$
for some sequence \smash{$\tilde{w}_i^{(k)}$}, $i=1,\ldots,d$, and
$k=1,2,3,\ldots$.  The $z$-updates in \eqref{eq:dyk_par} then become  
$$
X_i \tilde{w}_i^{(k)} =
(\Id-P_{C_i})(u_0^{(k)}+X_i\tilde{w}_i^{(k-1)}), 
\quad i=1,\ldots,d,
$$
and by Lemma \ref{lem:dyk_cd_update}, this is equivalent to 
$$
\tilde{w}_i^{(k)} = \argmin_{\tilde{w}_i \in \R^{p_i}} \;  
\half \|u_0^{(k)}+X_i\tilde{w}_i^{(k-1)} - X_i\tilde{w}_i\|_2^2 +
h_i(\tilde{w}_i), \quad i=1,\ldots,d.
$$
Rescaling once more, to $w_i^{(k)}=\gamma_i\tilde{w}_i^{(k)}$,
$i=1,\ldots,d$ and $k=1,2,3,\ldots$, gives the iterations
\eqref{eq:cd_dyk_par}. 

Lastly, we give a proof of Theorem \ref{thm:cd_dyk_par}.  We can write
the second set in the 2-set best approximation problem 
\eqref{eq:bap_prod2} as
$$
\tilde{C}_1 \times \cdots \times \tilde{C}_d 
= (M^T)^{-1} \Big( \tilde{D}_1 \times \cdots \times
\tilde{D}_d \Big),   
$$
where \smash{$\tilde{D}_i=\sqrt\gamma_i D_i$}, $i=1,\ldots,d$, and 
$$
M = \left(\begin{array}{cccc} 
X_1 & 0 & \ldots & 0 \\ 
0 & X_2 & \ldots & 0 \\
\vdots & & & \\
0 & 0 & \ldots & X_d \end{array}\right) \in \R^{nd \times p}. 
$$
The duality result established in Theorem
\ref{thm:dyk_cd_equiv} can now be applied directly to 
\eqref{eq:bap_prod2}.  (We note that the conditions of the theorem are
met because the matrix $M$, as defined above, has full column rank as
each $X_i$, $i=1,\ldots,d$ does.) Writing \smash{$h_S(v)=\max_{s \in S}
  \langle s,v \rangle$} for the support function of a set $S$, the
theorem tells us that the dual of \eqref{eq:bap_prod2} is
\begin{equation}
\label{eq:reg_prod}
\min_{\tilde{w} \in \R^p, \; \tilde\alpha \in \R^{nd}} \; 
\half \|\tilde{y}- M\tilde{w} - \tilde\alpha\|_2^2 + 
h_{\tilde{D}_1 \times \cdots \times \tilde{D}_d}(\tilde{w}) + 
h_{\tilde{C}_0}(\tilde\alpha),
\end{equation}
and the solutions in \eqref{eq:bap_prod2} and \eqref{eq:reg_prod}, 
denoted by \smash{$\tilde{u}^*$} and
\smash{$\tilde{w}^*,\tilde\alpha^*$} respectively, are related by
\begin{equation}
\label{eq:pd_relationship}
\tilde{u}_i^*=\sqrt\gamma_i y - X_i\tilde{w}_i^* - \tilde\alpha_i^*, 
\quad i=1,\ldots,d.
\end{equation}
Rescaling to
\smash{$(\bar{w}_1,\ldots,\bar{w}_d)=
(\tilde{w}_1/\sqrt\gamma_1,\ldots,\tilde{w}_d/\sqrt\gamma_d)$} 
and 
\smash{$(\bar\alpha_1,\ldots,\bar\alpha_d)=
(\tilde\alpha_1/\sqrt\gamma_1,\ldots,\tilde\alpha_d/\sqrt\gamma_d)$},  
the problem \eqref{eq:reg_prod} becomes 
\begin{align*}
&\min_{\bar{w} \in \R^p, \; \bar\alpha \in \R^{nd}} \;
\half \sum_{i=1}^d \gamma_i \|y - X_i\bar{w}_i - \bar\alpha_i \|_2^2 +    
\sum_{i=1}^d h_i (\gamma_i \bar{w}_i) +
       h_{C_0}(\gamma_1\bar\alpha_1,\ldots,\gamma_d\bar\alpha_d) \\  
\iff &\min_{\bar{w} \in \R^p, \; \bar\alpha \in \R^{nd}} \; 
\half \sum_{i=1}^d \gamma_i \|y - X_i\bar{w}_i - \bar\alpha_i \|_2^2 +    
\sum_{i=1}^d h_i (\gamma_i \bar{w}_i) \;\; \st \;\; \sum_{i=1}^d
       \gamma_i \bar\alpha_i = 0 \\  
\iff &\quad\;\, \min_{\bar{w} \in \R^p} \; 
\half \bigg\|y - \sum_{i=1}^d \gamma_i X_i\bar{w}_i \bigg\|_2^2 + 
       \sum_{i=1}^d h_i (\gamma_i\bar{w}_i).
\end{align*}
In the second line we rewrote the support function of
\smash{$D_1 \times \cdots \times D_d$} as a sum and that of $C_0$ as  
a constraint; in the third line we optimized over 
\smash{$\bar\alpha$} and used \smash{$\sum_{i=1}^d \gamma_i=1$}.
Clearly, the problem in the last display is exactly the 
regularized regression problem \eqref{eq:reg} after another rescaling, 
\smash{$(w_1,\ldots,w_d)=$}
\smash{$(\gamma_1\bar{w}_1,\ldots,\gamma_d\bar{w}_d)$}. 
That the solutions in \eqref{eq:bap_prod2}, \eqref{eq:reg_prod} are
related by \eqref{eq:pd_relationship} implies that the solutions  
\smash{$\hat{u},\hat{w}$} in \eqref{eq:bap_prod}, \eqref{eq:reg} are
related by  
$$
\hat{u}_1 = \cdots = \hat{u}_d = y - X\hat{w}.
$$

By Lemma 4.9 in \citet{han1988successive}, we know that when 
\eqref{eq:reg_prod} has a unique solution, the dual iterates in
Dykstra's algorithm for \eqref{eq:bap_prod2} converge to the solution
in \eqref{eq:reg_prod}.  Equivalently, when \eqref{eq:reg} has a
unique solution, the iterates \smash{$w^{(k)}$}, $k=1,2,3,\ldots$ in  
\eqref{eq:cd_dyk_par} converge to the solution in \eqref{eq:reg}. 

Also, by Theorem 4.7 in \citet{han1988successive}, if 
$$
\setint \bigcap_{i=1}^d (X_i^T)^{-1} (D_i) \not= \emptyset, 
$$ 
then the sequence \smash{$w^{(k)}$}, $k=1,2,3,\ldots$ produced by 
\eqref{eq:cd_dyk_par} has at least one accumulation point, and each
accumulation point solves \eqref{eq:reg}.  Moreover, the sequence 
\smash{$Xw^{(k)}$}, $k=1,2,3,\ldots$ converges to \smash{$X\hat{w}$},
the unique fitted value at optimality in \eqref{eq:reg}.
In fact, Theorem 4.8 in \citet{han1988successive} shows that a weaker
condition can be used when some of the sets are
polyhedral.  In particular, if $D_1,\ldots,D_q$ are polyhedral, then
the condition in the above display can be weakened to 
$$
(X_1^T)^{-1}(D_1) \cap \cdots \cap (X_q^T)^{-1}(D_q) \cap 
\setint (X_{q+1}^T)^{-1}(D_{q+1}) \cap \cdots \cap \setint
(X_d^T)^{-1} (D_d) \not= \emptyset, 
$$
and the same conclusion applies.

\section{Asymptotic linear convergence of the parallel-Dykstra-CD 
  iterations for the lasso problem}

Here we state and prove a result on the convergence rate of
the parallel-Dykstra-CD iterations \eqref{eq:cd_dyk_par} for the lasso
problem \eqref{eq:lasso}.  

\begin{theorem}[\textbf{Adaptation of \citealt{iusem1990convergence}}] 
\label{thm:cd_dyk_par_lasso_rate}
Assume the columns of $X \in \R^{n \times p}$ are in general position, 
and $\lambda>0$.
Then parallel-Dykstra-CD \eqref{eq:cd_dyk_par} for the lasso
\eqref{eq:lasso} has an asymptotically linear convergence rate, in
that for large enough $k$, using the notation of Theorem
\ref{thm:cd_lasso_rate_ip},
\begin{equation}
\label{eq:cd_dyk_par_lasso_rate}
\frac{\|w^{(k+1)} - \hat{w}\|_\Sigma}{\|w^{(k)}-\hat{w}\|_\Sigma} \leq  
\bigg(\frac{2a/\gamma_{\min}}{(2a/\gamma_{\min} + 
  \lambda_{\min}(X_A^T X_A)/ \max_{i \in A} \|X_i\|_2^2}\bigg)^{1/2},
\end{equation}
where $\gamma_{\min}=\min_{i=1,\ldots,p} \gamma_i \leq 1/p$ is the 
minimum of the weights.
\end{theorem}

We note that the parallel bound \eqref{eq:cd_dyk_par_lasso_rate} is 
worse than the serial bound \eqref{eq:cd_lasso_rate_ip}, because the
former relies on a quantity \smash{$2a/\gamma_{\min} \geq 2pa$} where
the latter relies on $a^2$. We conjecture that the bound
\eqref{eq:cd_dyk_par_lasso_rate} can be sharpened, by modifying the
parallel  algorithm so that we renormalize the weights in each cycle
after excluding the weights from zero coefficients.   

\begin{proof}
The proof is similar to that for Theorem
\ref{thm:cd_lasso_rate_ip}, given in Section
\ref{sec:cd_lasso_rate_ip}.
By Theorem 2 in \citet{iusem1990convergence}, for large enough $k$,
the iterates of \eqref{eq:dyk_par} satisfy
$$
\frac{\|u_0^{(k+1)} - \hat{u}\|_2}{\|u_0^{(k)} - \hat{u}\|_2} \leq  
\bigg(\frac{1}{1+\sigma}\bigg)^{1/2},
$$
where 
\smash{$\sigma = \mu^2/((1/\gamma_{\min}-1)(2-\gamma_{\min})) \geq  
\mu^2\gamma_{\min}/2$},
and $\mu>0$ is exactly as in Section 
\ref{sec:cd_lasso_rate_ip}.  The derivation details for
\eqref{eq:dyk_par}, \eqref{eq:cd_dyk_par} in the last section revealed
that the iterates from these two algorithms satisfy
$$
z_i^{(k)} = X_i w_i^{(k)} / \gamma_i, \quad i=1,\ldots,d  
\quad \text{and} \quad 
u_0^{(k+1)} = y - X w^{(k)},
\quad \text{for $k=1,2,3,\ldots$},
$$
hence
$$
\|u_0^{(k+1)}-\hat{u}\|_2 = \|Xw^{(k)}-X\hat{w}\|_2 =
\|w^{(k)}-\hat{w}\|_\Sigma, \quad \text{for $k=1,2,3,\ldots$}, 
$$
which gives the result.
\end{proof}

\section{Derivation details for \eqref{eq:cd_admm_par} and proof of
  Theorem \ref{thm:cd_admm_par}}
\label{sec:cd_admm_par}

Recall that \eqref{eq:bap_prod} can be rewritten as in
\eqref{eq:bap_prod2}.  The latter is a 2-set best approximation
problem, and so an ADMM algorithm takes the form of \eqref{eq:admm}
in Section \ref{sec:connections}.  Applying this to
\eqref{eq:bap_prod2}, and transforming the iterates back to their
original scale, we arrive at the following ADMM algorithm.  We
initialize 
\smash{$u_1^{(0)}=\cdots=u_d^{(0)}=0$},
\smash{$z_1^{(0)}=\cdots=z_d^{(0)}=0$},
and repeat for $k=1,2,3,\ldots$:
\begin{equation}
\label{eq:admm_par}
\begin{aligned}
&u_0^{(k)} = \frac{y}{1+\rho} + \frac{\rho}{1+\rho}
\sum_{i=1}^d \gamma_i(u_i^{(k-1)}-z_i^{(k-1)}), \\
&\begin{rcases*}
u_i^{(k)} = P_{C_i} (u_0^{(k)} + z_i^{(k-1)}), & \\
z_i^{(k)} = z_i^{(k-1)} + u_0^{(k)} - u_i^{(k)}, &
\end{rcases*}
\quad \text{for $i=1,\ldots,d$}.
\end{aligned}
\end{equation}
Basically the same arguments as those given in Section
\ref{sec:cd_dyk_par}, where we argued that \eqref{eq:dyk_par} is  
equivalent to \eqref{eq:cd_dyk_par}, now show that 
 \eqref{eq:admm_par} is equivalent to \eqref{eq:cd_admm_par}. 
Note that in the latter algorithm, we have slightly rewritten the
algorithm parameters, by using the notation
\smash{$\rho_i=\rho\gamma_i$}, $i=1,\ldots,d$.  That the
parallel-ADMM-CD iterations \eqref{eq:cd_admm_par} are equivalent to
the parallel-Dykstra-CD iterations \eqref{eq:cd_dyk_par} follows
from the equivalence of the 2-set Dykstra iterations 
\eqref{eq:dyk_par} and ADMM iterations \eqref{eq:admm_par}, which, 
recalling the discussion in Section \ref{sec:connections}, follows
from the fact that \smash{$\tilde{C}_0$} is a linear subspace and
\smash{$\tilde{y} \in \tilde{C}_0$} (i.e., $C_0$ is a linear subspace
and $(y,\ldots,y) \in C_0$).

The proof of Theorem \ref{thm:cd_admm_par} essentially just uses the
duality established in the proof of Theorem \ref{thm:cd_dyk_par} in
Section \ref{sec:cd_dyk_par}, and invokes standard theory for ADMM
from \citet{gabay1983applications,eckstein1992douglas,
boyd2011distributed}.  As shown previously, the dual of
\eqref{eq:bap_prod2} is \eqref{eq:reg_prod}, and by, e.g., the result
in Section 3.2 of \citet{boyd2011distributed}, which applies because   
$$
\|\tilde{y}-\tilde{u}\|_2^2, 1_{\tilde{C}_0}(\tilde{u}),
1_{\tilde{C}_1 \times \cdots \times \tilde{C}_d}(\tilde{u})
$$
are closed, convex functions of \smash{$\tilde{u}$}, the scaled dual
iterates in the ADMM algorithm for 
\eqref{eq:bap_prod2} converge to a solution in \eqref{eq:reg_prod}, or  
equivalently, the iterates \smash{$\rho_iz_i^{(k)}=X_iw_i^{(k)}$},
$i=1,\ldots,d$, $k=1,2,3,\ldots$ in \eqref{eq:admm_par} converge to
the optimal fitted values \smash{$X_i\hat{w}_i$}, $i=1,\ldots,d$ in  
\eqref{eq:reg}, or equivalently, the sequence \smash{$w^{(k)}$},
$k=1,2,3,\ldots$ in \eqref{eq:cd_admm_par} converges to a solution 
in \eqref{eq:reg}. 

\section{Numerical experiments comparing \eqref{eq:cd},
  \eqref{eq:cd_dyk_par}, \eqref{eq:cd_admm_par}} 

Figure \ref{fig:experiments} shows results from numerical
simulations comparing serial parallel coordinate descent
\eqref{eq:cd} to parallel-Dykstra-CD \eqref{eq:cd_dyk_par} and
parallel-ADMM-CD \eqref{eq:cd_admm_par} for the lasso problem. 
Our simulation setup was simple, and the goal was to investigate
the basic behavior of the new parallel proposals, and not to
investigate performance at large-scale nor compare to state-of-the
art implementations of coordinate descent for the lasso (ours was a 
standard implementation with no speedup tricks---like warm starts,
screening rules, or active set optimization---employed).

\begin{figure}[htbp]
\centering
\includegraphics[width=0.485\textwidth]{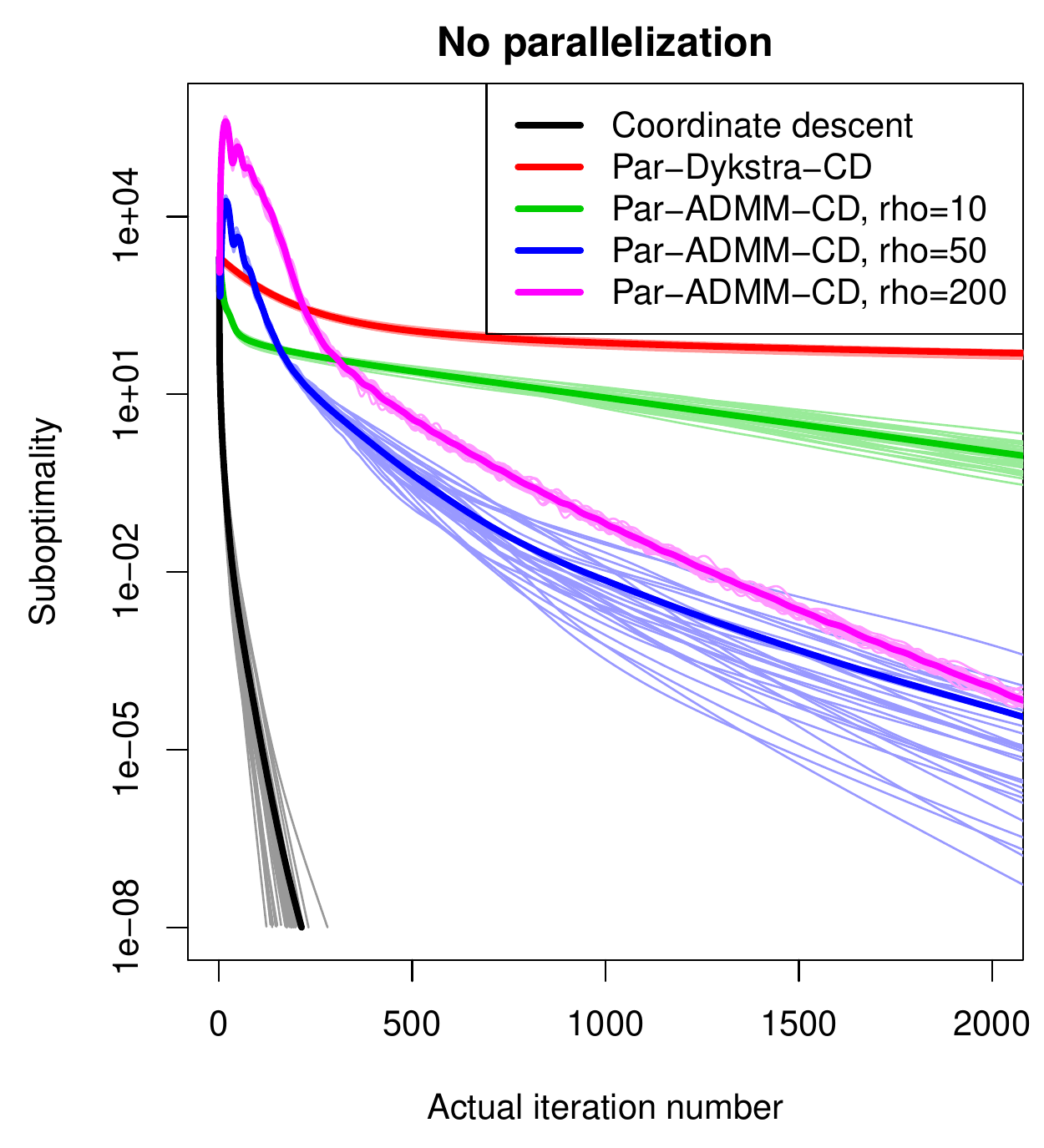}
\includegraphics[width=0.485\textwidth]{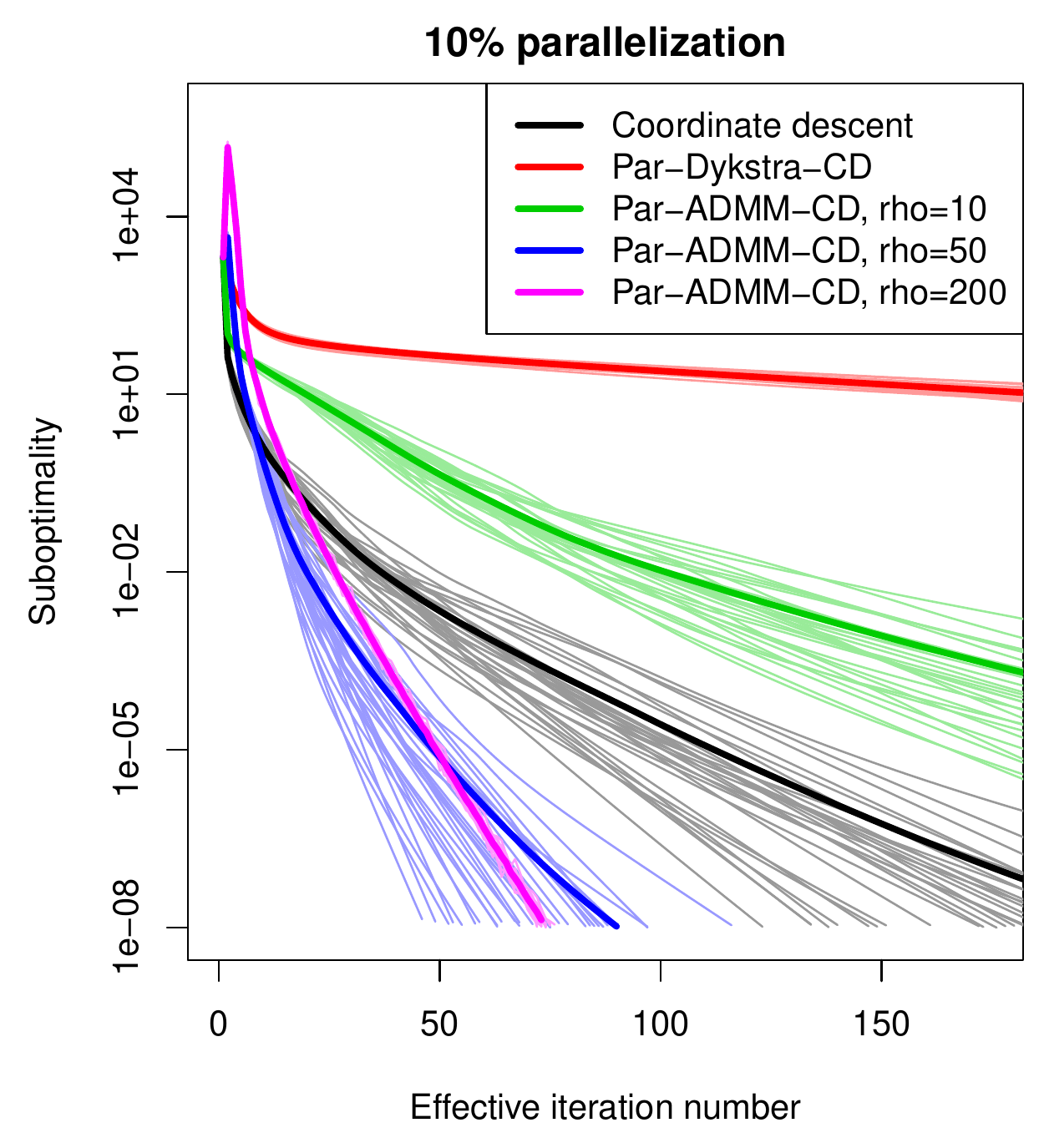}
\label{fig:experiments}
\caption{\it Suboptimality curves for serial coordinate descent,
  parallel-Dykstra-CD, and parallel-ADMM-CD, each run over the same 30 
lasso problems with $n=100$ and $p=500$. More details are given in
the text.}
\end{figure}

We considered a regression setting with $n=100$ observations and
$p=500$ predictors.  
Denoting by $x_i \in \R^p$ denotes the $i$th row of $X^{n \times p}$, 
the data was drawn according to the Gaussian linear model   
$$
x_i \sim N(0,I_{p\times p}) \quad \text{and}\quad  
y_i \sim N(x_i^T \beta_0, 1)
\quad \text{i.i.d., for $i=1,\ldots,n$},
$$
where $\beta_0 \in \R^p$ had its first 20 components equal to 1,
and the rest 0.  We computed solutions to the lasso problem
\eqref{eq:lasso} at $\lambda=5$, over 30 draws of data $X,y$ from the 
above model. At this value of $\lambda$, the lasso solution
\smash{$\hat{w}$} had an average of 151.4 nonzero components over the
30 trials. (Larger values of $\lambda$ resulted in faster convergence
for all algorithms and we found the comparisons more interesting at
this smaller, more challenging value of $\lambda$.)  

The figure shows the suboptimality, i.e., achieved criterion value 
minus optimal criterion value, as a function of iteration number, for: 
\begin{itemize}
\item the usual serial coordinate descent iterations \eqref{eq:cd}, in
  black; 
\item the parallel-Dykstra-CD iterations \eqref{eq:cd_dyk_par} with  
\smash{$\gamma_1=\cdots=\gamma_p=1/p$}, in red; 
\item the parallel-ADMM-CD iterations \eqref{eq:cd_admm_par} with 
\smash{$\rho_1=\cdots=\rho_p=1/p$} and 3 different
settings of \smash{$\rho=\sum_{i=1}^p \rho_i$}, namely
$\rho=10,50,200$, in green, blue, and purple respectively.  
\end{itemize}
(Recall that for $\rho=1$, parallel-ADMM-CD and parallel-Dykstra-CD
are equivalent.)  Thin colored lines in the figures denote the
suboptimality curves for individual lasso problem instances, and
thick colored lines represent the average suboptimality curves over
the 30 total instances.  In all instances, suboptimality is measured
with respect to the criterion value achieved by the least angle
regression algorithm \citep{efron2004least}, which is a direct
algorithm for the lasso and should return the exact solution up to
computer precision. 

The left panel of the figure displays the suboptimality curves as a
function of raw iteration number, which for the parallel methods
\eqref{eq:cd_dyk_par}, \eqref{eq:cd_admm_par} would correspond to
running these algorithms in a naive serial mode.  In the right panel,
iterations of the parallel methods are counted under a hypothetical 
``10\% efficient'' parallel implementation, where $0.1p$ updates of
the $p$ total updates in \eqref{eq:cd_dyk_par},
\eqref{eq:cd_admm_par} are able to be computed at the cost of 1 serial
update in \eqref{eq:cd}.  (A ``100\% efficient''
implementation would mean that all $p$ updates in
\eqref{eq:cd_dyk_par},  \eqref{eq:cd_admm_par} could be performed at
the cost of 1 serial update in \eqref{eq:cd}, which, depending on the 
situation, may certainly be unrealistic, due to a lack of 
available parallel processors, synchronization issues, etc.)  While
the parallel methods display much worse convergence based on raw
iteration number, they do offer clear benefits in the 10\% 
parallelized scenario.  Also, it seems that a larger value of $\rho$
generally leads to faster convergence, though the benefits of taking
$\rho=200$ over $\rho=50$ are not quite as clear (and for values of
$\rho$ much larger than 200, performance degrades).

\section{Proof of Theorem \ref{thm:dyk_cd_equiv_gen}}
\label{sec:dyk_cd_gen} 

First, we establish the following generalization of Lemma
\ref{lem:dyk_cd_update}. 

\begin{lemma}
\label{lem:dyk_cd_update_gen}
Let $f$ be a closed, strictly convex, differentiable function.
Assume $X_i \in \R^{n\times p_i}$ has full
column rank, and \smash{$h_i(v) =\max_{d \in D_i} \langle d,v 
  \rangle$} for a closed, convex set $D_i \subseteq \R^{p_i}$
containing 0. Then for \smash{$C_i=(X_i^T)^{-1}(D_i) \subseteq \R^n$}, 
and any $a \in \R^n$,  
$$
\hat{w}_i = \argmin_{w_i \in \R^{p_i}} \;
f(a+X_iw_i) + h_i(w_i) \iff
X_i \hat{w}_i = 
\Big( \nabla g - \nabla g \circ P_{C_i}^g \Big)\big(\nabla
g^*(-a)\big), 
$$
where $g(v)=f^*(-v)$.
\end{lemma}
\begin{proof}
We begin by analyzing the optimality condition that characterizes 
the Bregman projection 
\smash{$\hat{u}_i =P_{C_i}^g(x)
=\argmin_{c \in C_i} g(c) - g(x) - \langle \nabla g(x), c-x
\rangle$}, namely
$$
\nabla g(x) - \nabla g(\hat{u}_i) \in \partial 1_{C_i}(\hat{u}_i).
$$
Defining \smash{$\hat{z}_i=\nabla g(x) - \nabla g(\hat{u}_i)
=(\nabla g - \nabla g \circ P_{C_i}^g)(x)$},
this becomes
$$
\hat{z}_i \in \partial 1_{C_i}\Big( \nabla g^*\big(\nabla
g(x)-\hat{z}_i\big)\Big), 
$$
where we have used the fact that \smash{$\nabla g^*= (\nabla
  g)^{-1}$}, allowing us to rewrite the relationship between
\smash{$\hat{u}_i,\hat{z}_i$} as 
\smash{$\hat{u}_i=\nabla g^*(\nabla g(x)-\hat{z}_i)$}.  And lastly,
substituting $g(v)=f^*(-v)$ (and $g^*(v)=f(-v)$) the optimality
condition reads 
\begin{equation}
\label{eq:breg_resid}
\hat{z}_i \in \partial 1_{C_i}\Big( -\nabla f\big(\nabla
f^*(-x)+\hat{z}_i\big)\Big).
\end{equation}

Now we investigate the claim in the lemma. By subgradient optimality, 
\begin{align*}
\hat{w}_i = \argmin_{w_i \in \R^{p_i}} \;
f(a+X_iw_i) + h_i(w_i) &\iff
- X_i^T (\nabla f)(a+X_i\hat{w}_i) \in \partial h_i(\hat{w}_i) \\
&\iff
\hat{w}_i \in \partial h^*_i\Big( -X_i^T (\nabla f)(a+X_i\hat{w}_i) 
  \Big) \\
&\iff 
X_i \hat{w}_i \in X_i \partial h^*_i\Big( -X_i^T (\nabla f)
(a+X_i\hat{w}_i) \Big).
\end{align*}
The second line follows from the fact that, for a closed, convex 
function $g$, subgradients of $g$ and of $g^*$ are related via  
\smash{$x \in \partial f(y) \Longleftrightarrow y \in \partial
  g^*(x)$}; the third line follows from the fact that $X_i$ 
 has full column rank.  Note that
\smash{$h^*_i=1_{D_i}$}, the indicator function of $D_i$, and 
denote \smash{$h_{C_i}(v)=\sup_{c \in C_i} \langle c,v \rangle$}.
Then following from the last display, by the chain rule,
$$
\hat{w}_i = \argmin_{w_i \in \R^{p_i}} \;
f(a+X_iw_i) + h_i(w_i) \iff
X_i \hat{w}_i \in \partial h^*_{C_i}\Big(-\nabla f (a+X_i\hat{w}_i) 
\Big),  
$$
because \smash{$h^*_{C_i}=1_{C_i}=1_{D_i} \circ X_i^T$}. Applying the
previously established fact \eqref{eq:breg_resid} on Bregman
projections gives
$$
\hat{w}_i = \argmin_{w_i \in \R^{p_i}} \;
f(a+X_iw_i) + h_i(w_i) \iff
X_i \hat{w}_i = \Big(\nabla g - \nabla g \circ P_{C_i}^g\Big)(x) 
$$
for $a=\nabla f^*(-x)=-\nabla g(x)$, i.e., for $x=\nabla
g^*(-a)$.  This completes the proof of the lemma.
\end{proof} 

We are ready for the proof of the theorem.
We start with the claim about duality between \eqref{eq:reg_gen},   
\eqref{eq:bap_gen}. Standard arguments in convex analysis show that
the Lagrange dual of \eqref{eq:reg_gen} is
$$
\max_{u \in \R^n} \; -f^*(-u) - \sum_{i=1}^d h_i^*(X_i^T u),
$$
where $f^*$ is the conjugate of $f$ and \smash{$h_i^*=1_{D_i}$} 
the conjugate of $h_i$, $i=1,\ldots,d$, with the relationship
between the primal \smash{$\hat{w}$} and dual \smash{$\hat{u}$}
solutions being
\smash{$\hat{u} = -\nabla f( X\hat{w})$}.  Written in equivalent
form, the dual problem is
$$
\min_{u \in \R^n} \; f^*(-u) \quad \st \quad 
u \in C_1 \cap \cdots \cap C_d. 
$$
Recalling $g(v)=f^*(-v)$, and $b=-\nabla f(0)$, we have by
construction  
$$
D_g(u,b) = g(u) - g(b) - \langle \nabla g(b), u-b \rangle = 
f^*(-u) - f^*(\nabla f(0)),
$$
where we have used the fact that \smash{$\nabla g(b)=-\nabla
  f^*(\nabla f(0))=0$}, as \smash{$\nabla f^*=(\nabla
  f)^{-1}$}. Therefore the above dual problem, in the second to last
display, is equivalent to \eqref{eq:bap_gen}, establishing the claim.

Now we proceed to the claim about the equivalence between Dykstra's
algorithm \eqref{eq:dyk_gen} and coordinate descent
\eqref{eq:cd_gen}.  We note that a simple inductive argument shows
that the Dykstra iterates satisfy, for all $k=1,2,3,\ldots$,
$$
\nabla g(u_i^{(k)}) = -\sum_{j \leq i} z_j^{(k)}-\sum_{j > i}
z_j^{(k-1)}, \quad i=1,\ldots,d.
$$
We also note that, for $i=1,\ldots,d$, the image of \smash{$\nabla
  g - \nabla g \circ P_{C_i}^g$} is contained in the column span
of $X_i$.  To see this, write \smash{$\hat{u}_i=P_{C_i}^g(a)$}, and
recall the optimality condition for the Bregman projection,
$$
\langle \nabla g(\hat{u}_i) - \nabla g(a), c - \hat{u}_i \rangle \geq
0, \quad c \in C_i.
$$
If \smash{$\langle \nabla g(\hat{u}_i) - \nabla g(a), \delta \rangle 
  \not=0$} for some \smash{$\delta \in \nul(X_i^T)$},
supposing without a loss of generality that this inner product is  
negative, then the above optimality condition breaks for
\smash{$c=\hat{u}_i+\delta \in C_i$}.  Thus we have
shown by contradiction that \smash{$\nabla g(a) - \nabla g(\hat{u}_i)
  \perp \nul(X_i^T)$}, i.e., \smash{$\nabla g(a) - \nabla
  g(\hat{u}_i) \in \col(X_i)$}, the desired fact. 

Putting together the last two facts, we can write the  
Dykstra iterates, for $k=1,2,3,\ldots$, as
$$
z^{(k)}_i = X_i\tilde{w}_i^{(k)} \quad\text{and}\quad
\nabla g(u_i^{(k)}) = -\sum_{j \leq i}
X_j\tilde{w}_j^{(k)}-\sum_{j > i} X_j \tilde{w}_j^{(k-1)}, \quad 
\text{for $i=1,\ldots,d$},
$$
for some sequence \smash{$\tilde{w}_i^{(k)}$}, $i=1,\ldots,d$, and
$k=1,2,3,\ldots$.  In this parametrization, the $z$-updates in the
Dykstra iterations \eqref{eq:dyk_gen} are thus
$$
X_i \tilde{w}_i^{(k)} = \Big( \nabla g - \nabla g \circ P_{C_i}^g
\Big)\Bigg(\nabla g^*\bigg(-\sum_{j < i}
X_j\tilde{w}_j^{(k)}-\sum_{j > i} X_j \tilde{w}_j^{(k-1)}
\bigg)\Bigg), \quad i=1,\ldots,d,
$$
where we have used the fact that \smash{$\nabla g^*= (\nabla
  g)^{-1}$}. Invoking Lemma \ref{lem:dyk_cd_update_gen}, we know that
the above is equivalent to 
$$
\tilde{w}^{(k)}_i = \argmin_{\tilde{w}_i \in \R^{p_i}} \; 
f\bigg(\sum_{j < i} X_j\tilde{w}_j^{(k)} + \sum_{j > i} X_j 
\tilde{w}_j^{(k-1)} +X_i\tilde{w}_i\bigg) + h_i(\tilde{w}_i), 
\quad i=1,\ldots,d,
$$
which are exactly the coordinate descent iterations \eqref{eq:cd}.
It is easy to check that the initial conditions for the two
algorithms also match, and hence
\smash{$\tilde{w}_i^{(k)}=w_i^{(k)}$}, for all $i=1,\ldots,d$
and $k=1,2,3,\ldots$, completing the proof.  

\section{Parallel coordinate descent algorithms for
  \eqref{eq:reg_gen}} 

We first consider parallelization of projection algorithms for the best
Bregman-approximation problem \eqref{eq:bap_gen}.  
As in the Euclidean projection case, to derive parallel algorithms for 
\eqref{eq:bap_gen}, we will turn to a product space reparametrization,
namely, 
\begin{equation}
\label{eq:bap_gen_prod}
\min_{u \in \R^{nd}} \; D_{g^d}(u,\tilde{b})  
\quad \st \quad u \in C_0 \times (C_1 \times \cdots \times C_d), 
\end{equation}
where \smash{$C_0=\{ (u_1,\ldots,u_d) \in \R^{nd} :
  u_1=\cdots=u_d\}$}, 
\smash{$\tilde{b}=(b,\ldots,b) \in \R^{nd}$}, and we define the
function $g^d : \R^{nd} \to \R$ by
\smash{$g^d(u_1,\ldots,u_d)=\sum_{i=1}^d  
  g(u_i)$}.  For simplicity we have not introduced arbitrary
probability weights into \eqref{eq:bap_gen_prod}, and thus the 
resulting algorithms will not feature arbitrary weights for the
components, but this is a straightforward generalization and will
follow with only a bit more complicated notation.  

Dykstra's algorithm for the 2-set problem \eqref{eq:bap_gen_prod}
sets \smash{$u_1^{(0)}=\cdots=u_d^{(0)}=b$}, 
\smash{$r_1^{(0)}=\cdots=r_d^{(0)}=0$}, and
\smash{$z_1^{(0)}=\cdots=z_d^{(0)}=0$}, then repeats for
$k=1,2,3,\ldots$: 
\begin{equation}
\begin{aligned}
\label{eq:dyk_gen_par}
&u_0^{(k)} = \argmin_{u_0 \in \R^n} \; 
g(u_0) - \frac{1}{d} \sum_{i=1}^d \Big\langle 
\nabla g(u_i^{(k-1)}) + r_i^{(k-1)}, u_0 \Big \rangle, \\   
&\begin{rcases*}
r_i^{(k)} = \nabla g(u_i^{(k-1)}) + r_i^{(k-1)} - \nabla g(u_0^{(k)}),
& \\ 
u_i^{(k)} = (P^g_{C_i} \circ \nabla g^*)\Big(\nabla g (u_0^{(k)})
+ z_i^{(k-1)}\Big), & \\ 
z_i^{(k)} = \nabla g(u_0^{(k)}) + z_i^{(k-1)} - \nabla g(u_i^{(k)}), &  
\end{rcases*}
\quad \text{for $i=1,\ldots,d$}.
\end{aligned}
\end{equation}
Now we will rewrite the above iterations, under $b=-\nabla f(0)$,
where $g(v) =f^*(-v)$. 
The $u_0$-update in \eqref{eq:dyk_gen_par} is defined by the Bregman
projection of  
$$
\bigg( \nabla g^*\Big(\nabla g(u_1^{(k-1)})+r_1^{(k-1)}\Big),
\;\ldots,\; 
\nabla g^*\Big( \nabla g(u_d^{(k-1)})+r_d^{(k-1)} \Big) \bigg) 
\in \R^{nd} 
$$
onto the set $C_0$, with respect to the function $g^d$. By first-order
optimality, this update can be rewritten as
$$
\nabla g(u_0^{(k)}) = \frac{1}{d} \sum_{i=1}^d \Big(\nabla
g(u_i^{(k-1)}) + r_i^{(k-1)}\Big).
$$
Plugging in the form of the $r$-updates, and using a simple induction,
the above can be rewritten as   
\begin{align*}
\nabla g(u_0^{(k)}) &= \frac{1}{d}\sum_{i=1}^d \sum_{\ell=1}^{k-1}
\Big(\nabla g(u_i^{(\ell)}) - \nabla g(u_0^{(\ell)}) \Big) \\
&= \nabla g(b) + \frac{1}{d}\sum_{i=1}^d \sum_{\ell=1}^{k-1} 
(z_i^{(\ell-1)} - z_i^{(\ell)}) \\
&= -\frac{1}{d} \sum_{i=1}^d z_i^{(k)},  
\end{align*}
where in the second line we used the relationship given by
$z$-updates and recalled the initializations
\smash{$u_1^{(0)} = \cdots = u_d^{(0)} = b$}, and in the third line we
used \smash{$\nabla g(b)=\nabla g(-\nabla f(0))=\nabla g(\nabla
  g^*(-0))=0$}. 
Similar arguments as those given in the proof of
Theorem \ref{thm:dyk_cd_equiv_gen} in Section \ref{sec:dyk_cd_gen}
show that the $u$-updates and $z$-updates in \eqref{eq:dyk_gen_par}
can be themselves condensed 
to 
$$
X_i \tilde{w}_i^{(k)} = \Big( \nabla g - \nabla g \circ P_{C_i}^g
\Big)\bigg(\nabla g^* \Big(\nabla g (u_0^{(k)}) + X_i\tilde{w}_i^{(k)} 
\Big) \bigg), \quad i=1,\ldots,d,
$$
for a sequence \smash{$\tilde{w}_i^{(k)}$}, $i=1,\ldots,d$, and
$k=1,2,3,\ldots$, related by 
\smash{$z_i^{(k)}=X_i\tilde{w}_i^{(k)}$}, $i=1,\ldots,d$, and
$k=1,2,3,\ldots$.  By Lemma \ref{lem:dyk_cd_update_gen}, the above is
equivalent to
$$
\tilde{w}^{(k)}_i = \argmin_{\tilde{w}_i \in \R^{p_i}} \;
f\Big(-\nabla g (u_0^{(k)}) - X_i\tilde{w}_i^{(k)} + X_i\tilde{w}_i
\Big) + h_i(\tilde{w}_i), \quad i=1,\ldots,d,
$$
Rescaling to 
\smash{$w_i^{(k)}=\tilde{w}_i^{(k)}/d$}, $i=1,\ldots,d$, 
$k=1,2,3,\ldots$, the above displays show that the Dykstra
iterations \eqref{eq:dyk_gen_par} can be rewritten quite simply as: 
\begin{equation}
\label{eq:cd_dyk_gen_par}
w_i^{(k)} = \argmin_{w_i \in \R^{p_i}} \;
f(Xw^{(k)} - dX_iw_i^{(k)} + dX_iw_i) + h_i(dw_i), 
\quad i=1,\ldots,d, 
\end{equation}
for $k=1,2,3,\ldots$, with the initialization being
\smash{$w^{(0)}=0$}. This is our parallel-Dykstra-CD algorithm for
\eqref{eq:reg_gen}.  

Meanwhile, ADMM for the 2-set problem \eqref{eq:bap_gen_prod} 
is defined around the augmented Lagrangian
$$
L(u_0,\ldots,u_d,z_1,\ldots,z_d) = 
d g(u_0) -\langle \nabla g(b), u \rangle 
+ \sum_{i=1}^d \Big( \rho \|u_0-u_i+z_i\|_2^2 + 1_{C_i}(z_i)
\Big). 
$$
Initializing 
\smash{$u_1^{(0)}=\cdots=u_d^{(0)}=0$}, 
\smash{$z_1^{(0)}=\cdots=z_d^{(0)}=0$}, we repeat for
$k=1,2,3,\ldots$: 
\begin{equation}
\begin{aligned}
\label{eq:admm_gen_par}
&u_0^{(k)} = \argmin_{u_0 \in \R^n} \; 
g(u_0) -\langle \nabla g(b), u \rangle + \frac{\rho}{d} \sum_{i=1}^d  
\|u_0-u_i^{(k-1)}+z_i^{(k-1)}\|_2^2, \\
&\begin{rcases*}
u_i^{(k)} = P_{C_i} (u_0^{(k)}+ z_i^{(k-1)}), & \\ 
z_i^{(k)} = z_i^{(k-1)} + u_0^{(k)} - u_i^{(k)}, &  
\end{rcases*}
\quad \text{for $i=1,\ldots,d$}.
\end{aligned}
\end{equation}
Again using $b=-\nabla f(0)$, with $g(v) =f^*(-v)$, we will now
rewrite the above iterations. Precisely as in the connection between 
\eqref{eq:admm_par} and \eqref{eq:cd_admm_par} in the quadratic case,
as discussed in Section \ref{sec:cd_admm_par}, the $u$-updates and
$z$-updates here reduce to 
$$
w_i^{(k)} = \argmin_{\tilde{w}_i \in \R^{p_i}} \;  
\half \Big\|u_0^{(k)} + X_i \tilde{w}_i^{(k-1)} -
X_i \tilde{w}_i\Big\|_2^2 + h_i (\tilde{w}_i)
\quad i=1,\ldots,d,
$$
where \smash{$\tilde{w}_i^{(k)}$}, $i=1,\ldots,d$, 
$k=1,2,3,\ldots$ satisfies 
\smash{$z_i^{(k)}=X_i\tilde{w}_i^{(k)}$}, $i=1,\ldots,d$, 
$k=1,2,3,\ldots$.
The $u_0$-update here is characterized by
$$
\nabla g(u_0^{(k)}) = \frac{\rho}{d} \sum_{i=1}^d 
(u_i^{(k-1)}-z_i^{(k-1)}) -\rho u_0^{(k)},
$$
where we have used $\nabla g(b)=0$, or equivalently,  
$$
u_0^{(k)} = -\nabla f\bigg( 
\rho u_0^{(k)} -
\frac{\rho}{d} \sum_{i=1}^d  
(u_i^{(k-1)}-z_i^{(k-1)})\bigg),
$$
where we have used \smash{$\nabla g^*= (\nabla g)^{-1}$}
and $g^*(v)=f(-v)$, and lastly
$$
u_0^{(k)} = -\nabla f\bigg( 
\rho (u_0^{(k)} - u_0^{(k-1)}) -
\frac{\rho}{d} X(\tilde{w}^{(k-2)}-2\tilde{w}^{(k-1)}) \bigg),
$$
by plugging in the form of the $u$-updates, and recalling 
\smash{$z_i^{(k)}=X_i\tilde{w}_i^{(k)}$}, $i=1,\ldots,d$, 
$k=1,2,3,\ldots$.
Rescaling to \smash{$w_i^{(k)}=(\rho/d)\tilde{w}_i^{(k)}$},
$i=1,\ldots,d$, $k=1,2,3,\ldots$, and collecting the last several 
displays, we have shown that the ADMM
iterations \eqref{eq:admm_gen_par} can be written as:
\begin{equation}
\label{eq:cd_admm_gen_par}
\begin{gathered}
\text{Find $u_0^{(k)}$ such that:} \quad 
u_0^{(k)} = -\nabla f\Big(
\rho (u_0^{(k)} - u_0^{(k-1)}) -
X(\tilde{w}^{(k-2)}-2\tilde{w}^{(k-1)}) \Big), \\
w_i^{(k)} = \argmin_{w_i \in \R^{p_i}} \;  
\half \Big\|u_0^{(k)} + (d/\rho) X_i w_i^{(k-1)} -
(d/\rho) X_i w_i\Big\|_2^2 + h_i \big((d/\rho)w_i\big), 
\quad i=1,\ldots,d,
\end{gathered}
\end{equation}
for $k=1,2,3,\ldots$, where the initializations are
\smash{$u_0^{(0)}=0$}, \smash{$w^{(0)}=0$}.
This is our parallel-ADMM-CD algorithm for \eqref{eq:reg_gen}. 

Fortunately, many losses $f$ of interest in a regularized estimation
problem such as \eqref{eq:reg_gen} are separable and can be
expressed as \smash{$f(v)=\sum_{i=1}^n f_i(v_i)$} for smooth, convex 
$f_i$, $i=1,\ldots,n$.  This means that $\nabla_i f(v)= f_i'(v_i)$,
$i=1,\ldots,n$, and the $u_0$-update in
\eqref{eq:cd_admm_gen_par} reduces to $n$ univariate problems, each of
which can be solved efficiently via a simple bisection search. 

Note the stark contrast between the
parallel-Dykstra-CD iterations \eqref{eq:cd_dyk_gen_par} and
parallel-ADMM-CD iterations 
\eqref{eq:cd_admm_gen_par} in the general nonquadratic loss
case. These two methods are no longer equivalent for $\rho=1$ (or 
for any fixed $\rho$).  In 
\eqref{eq:cd_dyk_gen_par}, we perform (in parallel) a coordinatewise
$h_i$-regularized minimization involving $f$, for $i=1,\ldots,d$. In
\eqref{eq:cd_admm_gen_par}, we perform a single
quadratically-regularized minimization 
involving $f$ for the $u_0$-update, which in general   
may be difficult, but as explained above can be reduced to $n$
univariate minimizations for several losses $f$ of interest.  And for
the 
$w$-update, we perform (in parallel) a coordinatewise
$h_i$-regularized minimization involving a quadratic loss, for
$i=1,\ldots,d$, typically much cheaper than the analogous
minimizations for nonquadratic $f$.    

\bibliographystyle{plainnat}
\bibliography{ryantibs}

\end{document}